\documentclass[twocolumn]{autart}    % Enable this line and disable the 
\usepackage{graphicx}          % Include this line if your 
\usepackage{xcolor}
\usepackage{amsfonts}
\usepackage{amsmath}
\usepackage{amssymb}
\usepackage{mathrsfs}
\usepackage{float}
\usepackage{bm}
\usepackage{subcaption}
\usepackage{url}
\usepackage{diagbox}
\usepackage{mathtools}
\usepackage{cite}[sort,compress]

\usepackage{nicematrix}
\usepackage{arydshln}
\usepackage{multirow}

\usepackage{algorithm}
\usepackage{algpseudocode}

\begin{document}

\begin{frontmatter}

\title{A subspace approach to data-driven predictive control for linear parameter-varying systems\thanksref{footnoteinfo}} % Title, preferably not more 

\thanks[footnoteinfo]{Corresponding author Federico Porcari.}

\author[Polimi]{Federico Porcari}\ead{federico.porcari@polimi.it},
\author[UPenn,TUe]{Chris Verhoek}, %\ead{c.verhoek@tue.nl},
\author[TUe]{Valentina Breschi}, %\ead{v.breschi@tue.nl},
\author[TUe,SZTAKI]{Roland T\'oth}, %\ead{r.toth@tue.nl},
\author[Polimi]{Simone Formentin}%\ead{simone.formentin@polimi.it}

\address[Polimi]{Dipartimento di Elettronica, Informazione e Bioingegneria, Politecnico di Milano, P.za L. Da Vinci, 32, 20133 Milano, Italy.}  
\address[TUe]{Control Systems group, Dept. of Electrical Engineering, Eindhoven University of Technology, The Netherlands.}
\address[SZTAKI]{Systems and Control Lab, HUN-REN Institute for Computer Science and Control, 1111 Budapest, Hungary}
\address[UPenn]{Department of Electrical and Systems Engineering, University of Pennsylvania, United States}

                             % keyword list or with the 
                                          % help of the Automatica 
                                          % keyword wizard

\begin{abstract}                          % Abstract of not more than 
This paper presents a subspace data-driven predictive control method for linear parameter-varying (LPV) systems. Starting from an affine LPV state-space model in innovation form, we derive a multi-step predictor that separates the effects of past data, future inputs, scheduling trajectories, and innovations. By projecting this representation onto the row span of lifted input-output-scheduling data, we obtain an asymptotically unbiased data-driven predictor that can be embedded directly in a receding-horizon control problem, without explicitly identifying an LPV model. To make the resulting LPV data-driven predictive control (DDPC) formulation tractable, we introduce an LPV extension of $\gamma$-DDPC based on an LQ factorization. This formulation fixes the number of online decision variables independently of the length of the dataset. A reduced-order predictor is then proposed to curb the exponential growth of scheduling-dependent regressors, which also relaxes the persistence-of-excitation condition. Simulation studies, including an unbalanced-disk example, show that the proposed controller achieves good tracking performance and, compared to existing LPV DDPC schemes, achieves better robustness to measurement noise and reduced computational cost, making multi-step LPV DDPC practically deployable, even with longer past horizons.
\end{abstract}

\end{frontmatter}

\section{Introduction}

Linear parameter-varying (LPV) models provide a systematic way to represent nonlinear and/or time-varying dynamics through linear representations whose parameters depend on a measurable scheduling signal. In this way, the LPV framework bridges linear and nonlinear control, retaining much of the modeling, analysis, and synthesis machinery of the linear time-invariant (LTI) setting, while covering wide operating envelopes \cite{Cox:21a,Petreczky:17a,Toth:10a}. In parallel, the growing availability of informative process data has fueled the development of predictive control strategies that can be designed directly from measured trajectories, postponing or altogether avoiding an explicit model identification step whenever sufficiently informative data are available \cite{Coulson:19a,Verheijen:23a,Willems:05a}. Within this line of research, DeePC and related data-driven predictive control (DDPC) schemes have established a powerful paradigm: future trajectories are predicted directly from data and embedded into a receding-horizon optimal control problem \cite{Coulson:19a}. However, more recent results have shown that noise handling and computational scalability are key challenges in direct DDPC \cite{Berberich:20a,Breschi:23a,Favoreel:99a,Fiedler:21a,Mattsson:24a}.

Extending DDPC ideas from LTI systems to LPV systems is appealing, as LPV embeddings provide a structured route to data-driven control of nonlinear plants. At the same time, LPV data-driven representations must account for the interaction between inputs, outputs, and scheduling trajectories. This leads to lifted predictors whose dimensions, and thus the amount of data required to define them, grow rapidly with the scheduling dimension and with the past and prediction horizons. Recent work has begun to address the LPV DDPC problem from different directions. Direct LPV data-driven predictive control methods have been proposed on LPV behavioral arguments and LPV variants of Willems' Fundamental Lemma \cite{verhoek2021fundamental,Verhoek:25a,Verhoek:25b}, whereas alternative formulations based on polytopic LPV embeddings~\cite{BouHamdan:24a} have shown that LPV DDPC can also be realized without explicit forecasts of the future scheduling trajectory. In parallel, LPV subspace identification has developed several tools for estimating innovation-form predictors from data \cite{Cox:21a,vanWingerden:09a}, which can then be used inside indirect predictive controllers, i.e., that rely on an identified model \cite{Dong:09a}. These works strongly suggest that a fruitful route for LPV DDPC is to revisit direct data-driven prediction through the stochastic lens of LPV subspace identification, as purely deterministic data-driven predictors are typically fragile against noise. At the same time, such a route is viable only if the computational burden induced by LPV representations is explicitly addressed.

The present paper follows this route by formulating a computationally efficient, subspace direct LPV DDPC method that, unlike existing direct LPV DDPC approaches, explicitly accounts for noise. We consider data-generating systems representable as affine LPV state-space models in innovation form and derive a multi-step output predictor that explicitly separates the contribution of past measurements, future inputs, scheduling trajectories, and innovation noise. Projecting this predictor onto the row span generated by past data and future inputs yields an asymptotically unbiased predictor, with a similar rationale as LPV subspace identification, but without requiring an explicit model-estimation stage. 

This data-driven construction, however, presents computational limitations due to its LPV structure. Subspace predictors benefit from sufficiently long past horizons, as these horizons attenuate the effect of unknown initial conditions and improve prediction quality in noisy settings \cite{Breschi:23c,VanOverschee:96a}. In the LPV setting, increasing the past horizon also enlarges the lifted scheduling-dependent regressors and, consequently, the required data matrices and online optimization problem. This explains why existing LPV DDPC formulations often resort to short past horizons to remain computationally manageable \cite{Verhoek:25a}. While this choice is practical, it may limit noise rejection and prediction quality. 

Building on the $\gamma$-DDPC framework of \cite{Breschi:23a} and on the LPV data-driven control perspective of \cite{Verhoek:25a}, this paper develops a computationally efficient, subspace-inspired \emph{direct} data-driven predictive control method for LPV systems. The resulting controller does not rely on explicit identification of an LPV model and explicitly accounts for noise in the data. The main contributions are as follows

% Instead, it combines: \textit{i}) a stochastic LPV multi-step predictor derived from lifted data equations; \textit{ii}) an LQ-based reformulation that fixes the size of the online optimization variables independently of the dataset length; and \textit{iii}) a reduced-order approximation of the LPV predictor obtained by exploiting the contractive effect of normalized scheduling products together with a relevance-based selection of predictor rows. The computational ingredients are therefore not ancillary to the proposed method: the LQ reformulation removes the dependence of the online decision space on the number of collected samples, while the reduced-order construction efficiently reduces the constraint dimension. In this way, the proposed formulation directly accounts for the presence of noise in the data, while making LPV DDPC computationally viable for horizons and data lengths of practical interest.

% \paragraph*{Contributions.} The main contributions are as follows:
\begin{itemize}
    \item We derive a stochastic, subspace-inspired multi-step data-driven predictor for affine LPV systems in innovation form, establish its asymptotically unbiased approximation, and provide an explicit bound on the prediction error induced by relaxing exact Kronecker consistency of the lifted future input.
    \item We formulate an LPV $\gamma$-DDPC problem based on an LQ factorization of the data matrices. This reformulation fixes the dimension of the online decision variables independently of the number of collected samples.
    \item We propose a two-stage complexity-reduction strategy for LPV DDPC: \textit{i}) an \emph{a priori} pruning of high-order scheduling monomials after normalization of the scheduling signal; \textit{ii}) a relevance-based row-selection algorithm that retains only the most informative rows of the past and future data matrices. The resulting predictor is a reduced-order approximation of the full LPV predictor, enabling longer past horizons than those typically affordable in LPV DDPC formulations, while keeping the online problem dimension prescribed by user-selected complexity parameters.
    \item Through simulation studies, we illustrate improved robustness to measurement noise and a more favorable complexity/performance trade-off than the considered LPV-IO-DPC baseline \cite{Verhoek:25a}, while remaining close to an oracle LPV MPC benchmark.
\end{itemize}

\paragraph*{Outline.} The remainder of the paper is organized as follows. Section~\ref{sec:setting_goal} introduces the considered LPV data-driven control problem. Section~\ref{sec:data_to_predictions} derives the subspace-inspired predictor and its asymptotically unbiased data-driven approximation. Section~\ref{sec:limitations} formulates the corresponding LPV-DeePC scheme and discusses the structural and complexity limitations of the lifted LPV predictor, introducing the LQ-based compression used for an initial complexity reduction and the proposed LPV $\gamma$-DDPC formulation. Section~\ref{sec:intractability} introduces the tractable reduced form of the proposed LPV $\gamma$-DDPC and discusses how the proposed reductions make the approach computationally deployable. Section~\ref{sec:reduced_gamma} compares the proposed method with related LPV data-driven and subspace predictive control approaches. Sections~\ref{sec:hyperparameter_analysis} and~\ref{sec:unbalanced} demonstrate the efficiency and performance of the proposed approach through simulation studies, while Section~\ref{sec:conclusions} concludes the paper.

%\begin{itemize}
 %\item Motivation
 %\item Literature overview
 %\item Identify open challenges
 %\item Objective of the paper and sketch of proposed solution
 %\item Bullet-point-wise list of contributions
 %\item Short overview of content
%\end{itemize}

%\paragraph*{Contribution.} Building on the data-driven predictive control scheme recently proposed in \cite{Breschi:23a}, as well as on the rationale of \cite{Verhoek:25a}, we propose an alternative formulation for \emph{direct} data-driven predictive control for LPV systems. Our approach relies on a subspace-inspired reformulation of the data-driven predictor that allows us to %explicitly account for the properties of (process and measurement) noise. \textcolor{red}{More details to be added} 
%To enhance the approach's computational efficiency, we propose a complexity-reduction scheme that \textcolor{red}{More details to be added}. Unlike existing data-driven predictive control schemes for LPV systems, which scale with the number of available data, this design approach allows us to limit \textcolor{red}{how is it reduced should be added} the dimensions of the optimization problem to be solved at each time instant. As shown in our numerical experiments, our reformulation leads to improved noise rejection capailities in closed-loop with respect to the state-of-the-art and \textcolor{red}{More details to be added after reviewing the content.}  

%\paragraph*{Outline.}

\paragraph*{Notation.} The set of natural numbers excluding zero is denoted as $\mathbb{N}$, while $\mathbb{N}_0$ indicates the set of natural numbers including zero. Moreover, the sets of integer and real numbers are indicated as $\mathbb{Z}$ and $\mathbb{R}$, respectively. $I_n$ defines the identity matrix of size $n \times n$. Given a matrix $A \in \mathbb{R}^{n \times m}$, $A^\top$ denotes its transpose, while $\lVert A \rVert_2$ and $\lVert A \rVert_F$ are its 2-norm and Frobenius norm, respectively. Moreover, $A[i,:]$ indicates the $i$-th row of $A$, $A[:,i]$ denotes its $i$-th column, and $A[i:j,k:l]$ is its submatrix formed by the rows of $A$ from the $i$-th to the $j$-th, and its columns, from the $k$-th and $l$-th (with $i<j$ and $k<l$). Given another matrix $B \in \mathbb{R}^{p \times q}$, $A \otimes B \in \mathbb{R}^{np \times mq}$ is the Kronecker product between $A$ and $B$. The orthogonal projection of a matrix $A$ on the row span of a matrix $B$ is defined as $\Pi_B(A) = A B^\top (B B ^\top)^{-1} B$. For a vector $z \in \mathbb{R}^n$, $z^{[i]}$ denotes its $i$-th element while, given $A \in \mathbb{R}^{n\times n}$, $\lVert z \rVert_A=(z^\top A z)^{1/2}$. Given a signal $z_k \in \mathbb{R}^{n_\mathrm{z}}$, with $k \in \mathbb{Z}$, we compactly denote 
\begin{equation}
  z_{[k_1, k_2]}=\begin{bmatrix}
    z_{k_1}^{\top} & z_{k_1+1}^{\top} & \cdots & z_{k_2}^{\top}
\end{bmatrix}^{\top}\!,  
\end{equation}
while the Hankel matrix of width $N \in \mathbb{N}$ and depth $T \in \mathbb{N}$ associated with $z_{[k,k+T+N-2]}$ is defined as
\begin{equation}\label{eq:Hankel_def}
    Z_{k, T, N} = \frac{1}{\sqrt{N}} \begin{bmatrix}
        z_{k} & z_{k+1} & \cdots & z_{k+N-1} \\
        z_{k+1} & z_{k+2} & \cdots & z_{k+N} \\
        \vdots & \vdots & \ddots & \vdots \\
        z_{k+T-1} & z_{k+T} & \cdots & z_{k+T+N-2}
    \end{bmatrix}\!.
\end{equation}

\section{Setting \& Goal}\label{sec:setting_goal}
Consider the following discrete-time LPV system in state-space form (LPV-SS)
\begin{subequations}\label{eq:LPV-SS-original}
\begin{align}
    x_{k+1} & = \mathcal{A}(p_k) x_k + \mathcal{B}(p_k) u_k + \mathcal{K}(p_k) e_k, \label{eq:innov:state}\\
    y_k & = \mathcal{C}(p_k) x_k + \mathcal{D}(p_k) u_k + e_k,
\end{align}
where $x_{k} \!\in\! \mathbb{R}^{n_\textrm{x}}$, $u_{k} \!\in\! \mathbb{R}^{n_\textrm{u}}$, $y_k \!\in\! \mathbb{R}^{n_\textrm{y}}$, and $p_k \!\in\! \mathbb{P} \!\subset\! \mathbb{R}^{n_\textrm{p}}$ are the state, input, output, and scheduling signal at time $k \in \mathbb{Z}$, respectively, and $\mathbb{P}$ is a compact set. Meanwhile, $e_{k} \in \mathbb{R}^{n_\textrm{y}}$ is the realization at time $k$ of an i.i.d. zero-mean white innovation process with variance $\sigma_e^2$. Let us consider the state-space matrices to be affine in the scheduling signal, i.e.,
\begin{equation}
    \begin{aligned}
        & \mathcal{A}(p_k)\!=\!A_0 \!+\! \sum_{i=1}^{n_\textrm{p}} p_k^{[i]} A_i,~~\mathcal{B}(p_k) \!=\! B_0 \!+\! \sum_{i=1}^{n_\textrm{p}} p_k^{[i]} B_i,\\
        & \mathcal{K}(p_k)= K_0 + \sum_{i=1}^{n_\textrm{p}} p_k^{[i]} K_i,~~\mathcal{C}(p_k)\!=\!C_0 \!+\! \sum_{i=1}^{n_\textrm{p}} p_k^{[i]} C_i,\\
        & \mathcal{D}(p_k) \!=\! D_0 \!+\! \sum_{i=1}^{n_\textrm{p}} p_k^{[i]} D_i,
    \end{aligned} 
\end{equation}
\end{subequations}
with $\{A_{i},B_{i},K_{i},C_{i},D_{i}\}_{i=1}^{n_\textrm{p}}$ being real, yet \emph{unknown}, matrices characterizing the dynamics of the system. Suppose that the LPV-SS system is \emph{controllable} and \emph{observable} according to the definitions provided in~\cite{Petreczky:17a,Verhoek:25b}, and that the matrices of its state-space model satisfy the following assumption.

\begin{assum}[Stable innovation dynamics] \label{assum:stable_noise_dynamics}\hspace{-0pt}
    By substituting $e_k=y_k-\mathcal{C}(p_k) x_k - \mathcal{D}(p_k) u_k$ into \eqref{eq:innov:state}, the resulting so-called closed-loop innovation dynamics are exponentially stable w.r.t. to the origin under any scheduling sequence $p_k\in \mathbb{P}$ with $k\geq0$ and initial condition $x_0\in\mathbb{R}^{n_\mathrm{x}}$. 
    %scheduling dependent state transition matrix $\bar{\mathcal{A}}(p) = \mathcal{A}(p) - \mathcal{K}(p) \mathcal{C}(p)$ is exponentially stable.
\end{assum}

Note that this is a common assumption for LPV-SS identification (see, e.g., \cite{Cox:21a,Verdult:02a}).

Let us assume that we have excited the system in open-loop and collected the probing inputs, as well as the corresponding outputs and scheduling signals, thus having access to a dataset
\begin{equation}\label{eq:dataset}
    \mathfrak{D}_{N_{\textrm{data}}} = \{u_k, y_k, p_k \}_{k=0}^{N_\textrm{data}-1}.
\end{equation}
Under these assumptions, our goal is to use the available data in \eqref{eq:dataset} to design a \emph{computationally efficient} predictive controller for the system to track a user-defined reference $y^{r}_{k}$, for $k \in \mathbb{N}_0$, \emph{without} performing any explicit identification step.

\section{From data to data-driven predictions}\label{sec:data_to_predictions}
We now leverage ideas from subspace identification (see~\cite{Cox:21a,VanOverschee:96a}) to define a data-based multistep output predictor for \eqref{eq:LPV-SS-original}, that will be later used to formulate the Data-Driven Predictive Control (DDPC) problem. To this end, let us define the scheduling-independent matrices
\begin{subequations}
\begin{equation}\label{eq:scheduless_matrix}
    \begin{aligned}
        A & = \begin{bmatrix}
            A_0 & \cdots & A_{n_\textrm{p}}
        \end{bmatrix},
        & B & = \begin{bmatrix}
            B_0 & \cdots & B_{n_\textrm{p}}
        \end{bmatrix}, \\
        K & = \begin{bmatrix}
            K_0 & \cdots & K_{n_\textrm{p}}
        \end{bmatrix}, & C & = \begin{bmatrix}
            C_0 & \cdots & C_{n_\textrm{p}}
        \end{bmatrix},\\
        D & = \begin{bmatrix}
            D_0 & \cdots & D_{n_\textrm{p}}
        \end{bmatrix},
        & \tilde{I} & = \begin{bmatrix}
            I_{n_\textrm{y}} & 0 &\cdots & 0
        \end{bmatrix},
    \end{aligned}
\end{equation}
and the extended scheduling signal 
\begin{equation}\label{eq:extended_scheduling}
q_k=\begin{bmatrix}
    1 & p_k^{\top}
    \end{bmatrix}^{\top} \in \mathbb{R}^{n_\textrm{q}},~~~\mbox{with}~~n_\textrm{q} = n_\textrm{p} + 1.
\end{equation}
Accordingly, we can equivalently recast \eqref{eq:LPV-SS-original} as
\begin{align}
        \label{eq:LPV-SS-state} x_{k+1} & = A (q_k \otimes x_k) + B (q_k \otimes u_k) + K (q_k \otimes e_k), \\
        \label{eq:LPV-SS-output} y_k & = C (q_k \otimes x_k) + D (q_k \otimes u_k) + \tilde{I}(q_k \otimes e_k).
\end{align}
\end{subequations}
Exploiting this equivalent representation, the $t$-step ahead output of the system can be characterized by recursively evaluating %iterating 
the state equation over such a prediction horizon. Specifically, given an initial condition $x_k$, an input sequence $u_{[k, k+t]}$, an innovation sequence $e_{[k, k+t]}$, and a scheduling signal realization $p_{[k, k+t]}$ (either known or estimated), the $t$-step ahead output can be computed as
\begin{subequations} 
\begin{equation}\label{eq:future_prediction}
    y_{k+t} = \Gamma_t x_{k,t}^F + \mathcal{H}_{\mathrm{d},t} u_{k,t}^F + \mathcal{H}_{\mathrm{s},t} e_{k,t}^F, 
\end{equation}
where 
\begin{align}
&  x_{k,t}^F = q_{k+t} \otimes q_{k+t-1} \otimes \ldots \otimes q_k \otimes x_k,\\
& u_{k,t}^F  = \begin{bmatrix}
        q_{k+t} \otimes q_{k+t-1} \otimes \cdots \otimes q_k \otimes u_k \\
        q_{k+t} \otimes q_{k+t-1} \otimes \cdots \otimes q_{k+1} \otimes u_{k+1} \\
        \vdots \\
        q_{k+t} \otimes u_{k+t}
    \end{bmatrix}, \label{eq:definition-of-uF}\\
&   e_{k,t}^F  = \begin{bmatrix}
        q_{k+t} \otimes q_{k+t-1} \otimes \cdots \otimes q_k \otimes e_k \\
        q_{k+t} \otimes q_{k+t-1} \otimes \cdots \otimes q_{k+1} \otimes e_{k+1} \\
        \vdots \\
        q_{k+t} \otimes e_{k+t}
    \end{bmatrix},
\end{align}
\end{subequations}
$\mathcal{H}_{\mathrm{d},0} = D$, $\mathcal{H}_{\mathrm{s},0} = \tilde{I}$, and
\begin{equation}\label{eq:useful_matrices2}
\begin{aligned}
     &\Gamma_t \!=\! C \displaystyle \!\prod_{\tau=1}^{t}\! 
     \mathtt{A}_\tau^{n_\textrm{q}},~~\mathcal{H}_{\mathrm{d},t}\!=\!\!\left[\begin{array}{@{}c;{2pt/2pt}c@{}}
        C \!\left( \displaystyle \prod_{\tau=1}^{t-1} \mathtt{A}_\tau^{n_\textrm{q}} \right) \mathtt{B}_t^{n_\textrm{q}} \!& 
        \mathcal{H}_{\mathrm{d},t-1}
    \end{array}\right]\!,\\ 
    &\qquad \quad \quad \mathcal{H}_{\mathrm{s},t}\!=\!\!\left[\begin{array}{@{}c;{2pt/2pt}c@{}}
        C \!\left( \displaystyle \prod_{\tau=1}^{t-1} \mathtt{A}_\tau^{n_\textrm{q}} \right) \mathtt{K}_t^{n_\textrm{q}} \!&
        \mathcal{H}_{\mathrm{s},t-1}
    \end{array}\right],
\end{aligned}    
\end{equation}
with 
\begin{equation*}
    \mathtt{A}_\tau^{n}=I_{n^\tau}\otimes A,~~\mathtt{B}_\tau^{n}=I_{n^\tau}\otimes B,~~~\mathtt{K}_\tau^{n}=I_{n^\tau}\otimes K.
\end{equation*}
We then exploit \eqref{eq:LPV-SS-output} to rewrite the innovation as 
\begin{equation}\label{eq:error_dynamics}
    e_k = y_k - C (q_k \otimes x_k) - D (q_k \otimes u_k).
\end{equation}
By substituting this relation into \eqref{eq:LPV-SS-state}, we can then rewrite the state $x_k$ as
\begin{equation*}
        x_k \!\!=\!\!\bar{A} (\zeta_{k-\!1}\!\otimes x_{k-\!1})\!+\! \bar{B}(\zeta_{k-\!1}\!\otimes u_{k-\!1})\!+\! K (q_{k-\!1} \!\otimes\!y_{k-\!1}),
\end{equation*}% <- to avoid an extra space before 'where' in the text
\begin{table*}[!t]
\begin{equation}\label{eq:matrices_defined}
\begin{aligned}
    & \bar{A}\!=\!\!\begin{bmatrix}
        A_0\!-\!\! K_0 C_0 & A_1\!-\!\! [K\!C]_{0}^{1} \!&\! \cdots \!\!&\! A_{n_\textrm{p}} \!\!\!-\!\! [K\!C]_{0}^{n_\textrm{p}} \!&\! -\!K_1 C_1 \!&\! -\! [K\!C]_{1}^{2} \!&\! \cdots \!&\! -\![K\!C]_{1}^{n_\textrm{p}} \!&\! -\!K_2 C_2 \!&\! -\![K\!C]_{2}^{3} \!\;\;\! \cdots \!\;\;\!\end{bmatrix}\!,\mbox{ with } [K\!C]_{i}^{j}\!=\!K_{i}C_{j}\!+\!K_{j}C_{i},\\
    & \bar{B}\!=\!\!\begin{bmatrix} B_0 \!\!-\!\! K_0 D_0 \!&\! B_1 \!-\! [K\!D]_{0}^{1} \!&\!\! \cdots \!\!& \!\! B_{n_\textrm{p}} \!\!-\! [K\!D]_{0}^{n_\textrm{p}} \!&\! \!-\!K_1 D_1 \!&\! \!-\!  [K\!D]_{1}^{2} \!&\!\! \cdots \!\!& \!\!-\![K\!D]_{1}^{n_\textrm{p}} \!&\! \!-\!K_2 D_2 \!&\! \!-\![K\!D]_{2}^{3}  \!\;\;\! \cdots\!\! \;\; \end{bmatrix}\!,\mbox{ with } [K\!D]_{i}^{j}\!=\!K_{i}D_{j}\!+\!K_{j}D_{i}.
\end{aligned}
\end{equation}
\begin{equation}\label{eq:zeta_def}
    \zeta_k = \begin{bmatrix}
        1 & p_k^{[1]} & \cdots &  p_k^{[n_\textrm{p}]} & p_k^{[1]} p_k^{[1]} & \cdots & p_k^{[1]} p_k^{[n_\textrm{p}]}  & p_k^{[2]} p_k^{[2]} &  & p_k^{[2]} p_k^{[3]} & \cdots \;\; p_k^{[n_\textrm{p}]} p_k^{[n_\textrm{p}]}
    \end{bmatrix}^\top,
\end{equation}
\hrule
\end{table*}% <- to avoid an extra space before 'where' in the text
where $\bar{A}$ and $\bar{B}$ are defined as in \eqref{eq:matrices_defined} on Page~\pageref{eq:matrices_defined} and $\zeta_k \in \mathbb{R}^{n_\zeta}$ comprises all unique combinations in $q_k\otimes q_k$ (see \eqref{eq:zeta_def} on Page~\pageref{eq:zeta_def}), with $n_\zeta = \frac{n_\textrm{q}(n_\textrm{q}+1)}{2}$. Using this relation iteratively, it is straightforward to see that $x_k$ can be reconstructed from an $M$-long past trajectory $\{u_{[k-M, k-1]},y_{[k-M, k-1]},p_{[k-M, k-1]}\}$ and the state $x_{k-M}$ as 
\begin{subequations}\label{eq:past_prediction}
\begin{equation}
    x_k = \Lambda_M x_{k,M}^P + \Phi_M u_{k,M}^P + \Psi_M y_{k,M}^P,
\end{equation}
where
\begin{align}
         &x_{k,M}^P = \zeta_{k-1} \otimes \zeta_{k-2} \otimes \cdots \otimes \zeta_{k-M} \otimes x_{k-M}, \\
    & u_{k,M}^P \!=\!\! \begin{bmatrix}
        \zeta_{k-1} \otimes \cdots \otimes \zeta_{k-M} \otimes u_{k-M} \\
        \zeta_{k-1} \otimes \cdots \otimes \zeta_{k-M+1} \otimes u_{k-M+1} \\
        \vdots \\
        \zeta_{k-1} \otimes u_{k-1}
    \end{bmatrix}\!, \\
    & y_{k,M}^P \!=\!\! \begin{bmatrix}
        \zeta_{k-1} \otimes \cdots \otimes \zeta_{k-M+1} \otimes q_{k-M} \otimes y_{k-M} \\
        \zeta_{k-1} \otimes \cdots \otimes \zeta_{k-M+2} \otimes q_{k-M+1} \otimes y_{k-M+1} \\
        \vdots \\
        q_{k-1} \otimes y_{k-1}
    \end{bmatrix},
\end{align}
\end{subequations}
$\Phi_1 = \bar{B}$, $\Psi_1 = K$, and 
\begin{equation}\label{eq:useful_matrices}
    \begin{aligned}
    &\Lambda_M  = \prod_{i=0}^{M-1} \bar{\mathtt{A}}_i^{n_\zeta},
    ~~~\Phi_M  = \left[\begin{array}{@{}c;{2pt/2pt}c@{}}
        \left( \displaystyle \prod_{i=0}^{M-2} \bar{\mathtt{A}}_i^{n_\zeta} \! \right) \bar{\mathtt{B}}_{M-1}^{n_\zeta} &
        \Phi_{M-1}
    \end{array}\right], \\
    & \qquad \qquad \Psi_M  = \left[\begin{array}{@{}c;{2pt/2pt}c@{}}
        \left( \displaystyle \prod_{i=0}^{M-2} \bar{\mathtt{A}}_i^{n_\zeta} \! \right)\mathtt{K}_{M-1}^{n_\zeta} &
        \Psi_{M-1}
    \end{array}\right], 
\end{aligned}
\end{equation}
with
\begin{equation*}
    \bar{\mathtt{A}}_i^{n}=I_{n^i}\otimes \bar{A},\qquad \bar{\mathtt{B}}_i^{n}=I_{n^i}\otimes \bar{B}.
\end{equation*}
Substituting \eqref{eq:past_prediction} into \eqref{eq:future_prediction}, the $t$-step ahead output can ultimately be recast as
\begin{subequations} \label{eq:single_predictor}
\begin{equation}
    y_{k+t} = \Gamma_{x,t} x_{k,M,t}^0 + \Gamma_{z,t} z_{k,M,t}^P + \mathcal{H}_{\mathrm{d},t} u_{k,t}^F + \mathcal{H}_{\mathrm{s},t} e_{k,t}^F,
\end{equation}
where
\begin{align}
    & \Gamma_{x,t} = \Gamma_t \left(I_{n_\textrm{q}^{t+1}} \otimes \Lambda_M \right), \\
    & \Gamma_{z,t} = \Gamma_t \left(I_{n_\textrm{q}^{t+1}} \otimes \begin{bmatrix}
        \Phi_M & \Psi_M
    \end{bmatrix}\right), \\
    & x_{k,M,t}^{0} = q_{k+t} \otimes q_{k+t-1} \otimes \ldots \otimes q_k \otimes x^P_{k,M}, \\
    & z_{k,M,t}^{P} = q_{k+t} \otimes q_{k+t-1} \otimes \ldots \otimes q_k \otimes \begin{bmatrix}
        u^{P}_{k,M} \\ y^{P}_{k,M}
    \end{bmatrix}, \label{eq:definition-of-zP}
\end{align}
where $\Gamma_t$ is defined as in \eqref{eq:useful_matrices2}, and $\Lambda_M$, $\Phi_M$, $\Psi_M$ are introduced in \eqref{eq:useful_matrices}. 
\end{subequations}

While \eqref{eq:single_predictor} returns only the $t$-step ahead output, this relationship allows us to build a multistep output predictor over a prediction horizon $T \in \mathbb{N}$. In particular, by defining
\begin{equation*}
    \begin{aligned}
        & \Gamma_x = \begin{bmatrix}
            \begin{bmatrix}
                 \Gamma_{x,0} & 0 & \cdots & 0
            \end{bmatrix} \\ 
            \begin{bmatrix}
                 \Gamma_{x,1} & \cdots & 0
            \end{bmatrix} \\ 
            \vdots \\ 
            \begin{bmatrix}
                \Gamma_{x,T-2} & 0
            \end{bmatrix} \\ 
            \Gamma_{x,T-1}
        \end{bmatrix}, \quad 
        & \Gamma_z = \begin{bmatrix}
            \begin{bmatrix}
                 \Gamma_{z,0} & 0 & \cdots & 0
            \end{bmatrix} \\ 
            \begin{bmatrix}
                 \Gamma_{z,1} & \cdots & 0
            \end{bmatrix} \\ 
            \vdots \\ 
            \begin{bmatrix}
                \Gamma_{z,T-2} & 0
            \end{bmatrix} \\ 
            \Gamma_{z,T-1}
        \end{bmatrix}, \\
        & \mathcal{H}_\mathrm{d} = \begin{bmatrix}
            \begin{bmatrix}
                \mathcal{H}_{\mathrm{d},0} & 0 & \cdots & 0
            \end{bmatrix} \\ 
            \begin{bmatrix}
                \mathcal{H}_{\mathrm{d},1} & \cdots & 0
            \end{bmatrix} \\
            \vdots \\ 
            \begin{bmatrix}
                \mathcal{H}_{\mathrm{d},T-2} & 0
            \end{bmatrix} \\ 
            \mathcal{H}_{\mathrm{d},T-1} 
        \end{bmatrix}, \quad 
        & \mathcal{H}_\mathrm{s} = \begin{bmatrix}
            \begin{bmatrix}
                \mathcal{H}_{\mathrm{s},0} & 0 & \cdots & 0
            \end{bmatrix} \\ 
            \begin{bmatrix}
                \mathcal{H}_{\mathrm{s},1} & \cdots & 0
            \end{bmatrix} \\
            \vdots \\ 
            \begin{bmatrix}
                \mathcal{H}_{\mathrm{s},T-2} & 0
            \end{bmatrix} \\ 
            \mathcal{H}_{\mathrm{s},T-1} 
        \end{bmatrix}, 
    \end{aligned}
\end{equation*}
and omitting the past horizon $M$ and the future horizon $t = T-1$ to simplify notation, i.e., $x_k^0 \equiv x^0_{k,M,t}$, $z_k^P \equiv z^P_{k,M,t}$, $u_k^F \equiv u^F_{k,t}$, and $e_k^F \equiv e^F_{k,t}$, we obtain: %one gets:
\begin{align}\label{eq:multiple_predictor}
    y_{[k,k+T-1]} = \Gamma_x x_{k}^0 + \Gamma_z z_{k}^{P}+\mathcal{H}_\mathrm{d} u_{k}^{F}+\mathcal{H}_\mathrm{s} e_{k}^{F}.
\end{align}
The same relationship holds for the data in $\mathfrak{D}_{N_{\textrm{data}}}$. Specifically, by introducing the Hankel matrices \eqref{eq:Hankel_def} of ``past'' and ``future'' signals, i.e., 
\begin{align}\label{eq:Hankel_matrices}
    X_0 & = X^0_{k,1,N},
    && Z_P = Z^P_{k,1,N},\\
    Y_F & = Y_{k, T, N}, 
    && U_F = U^F_{k,1,N}, \quad
    E_F = E^F_{k,1,N}, \nonumber
\end{align}
with $N = N_\textrm{data} - M - T + 1$, then the following holds: 
\begin{equation}\label{eq:true_predictor}
    Y_F = \Gamma_x X_0 + \Gamma_z Z_P + \mathcal{H}_\mathrm{d} U_F + \mathcal{H}_\mathrm{s} E_F.
\end{equation}
As both the initial state $X_0$ and the future innovation noise $E_F$ are unavailable in practice, this (exact) representation cannot yet be used to build a data-based predictor. Nonetheless, it can be recovered asymptotically by projecting $Y_F$ on the row span of $\begin{bmatrix}
        Z_P^\top \; U_F^\top\end{bmatrix}$, as formalized in the following lemma.
\begin{lem}[Asymptotic predictor] \label{lem:asymptotic_predictor_lemma}
    Let Assumption~\ref{assum:stable_noise_dynamics} hold. Then, for $N \to \infty$ and $M \to \infty$ the projection $\hat{Y}_F = \Pi_{Z_P, U_F}(Y_F)$ of the multistep predictor $Y_F$ on the row span of $\begin{bmatrix}
        Z_P^\top \; U_F^\top\end{bmatrix}$, i.e.,
         \begin{equation}\label{eq:asymptotic_predictor}
        \hat{Y}_F = \Gamma_z Z_P + \mathcal{H}_\mathrm{d} U_F,
    \end{equation}    
    is an asymptotically unbiased estimate of $Y_F$.
\end{lem}
\begin{pf}
    Under Assumption~\ref{assum:stable_noise_dynamics}, there exists a finite constant $c > 0$ and a constant $0 \leq \lambda < 1$ such that
    \begin{align*}
        \left \lVert \Lambda_M x_{k,M}^P \right \rVert_2  
        & =\! \left \lVert \left( \prod_{i=1}^{M} \bar{\mathcal{A}}(p_{k-i}) \right) x_{k-M} \right \rVert_2\!\!\!\leq c \lambda^M \left \lVert x_{k-M} \right \rVert_2.
    \end{align*}
     Hence, for $M \to \infty$, $\lVert \Lambda_M x_{k,M}^P \rVert_2$ converges exponentially to zero under any $p_{k-i}\in\mathbb{P}$. This implies that also $\lVert \Gamma_x X_0 \rVert_2$ converges exponentially to zero and the effect of the initial state $X_0$ in \eqref{eq:true_predictor} vanishes for $M \to \infty$. The same applies for the projection $\hat{X}_0 = \Pi_{Z_P,U_F}(X_0)$, as 
     \begin{equation*}
         \lVert \Gamma_x \hat{X}_0 \rVert_2 = \left\lVert \Gamma_x X_0 \begin{bmatrix}
            Z_P \\ U_F
        \end{bmatrix}^\dagger \begin{bmatrix}
            Z_P \\ U_F
        \end{bmatrix} \right\rVert_2 \leq \lVert \Gamma_x X_0 \rVert_2, 
     \end{equation*}
     where $\dagger$ denotes the Moore–Penrose inverse. Meanwhile, by definition, the projection $\hat{E}_F = \Pi_{Z_P, U_F}(E_F)$ of the future innovations is:   
    \begin{equation*}
        \hat{E}_F = E_F \begin{bmatrix}
            Z_P^\top & U_F^\top
        \end{bmatrix} \left( \begin{bmatrix}
            Z_P \\ U_F
        \end{bmatrix} \begin{bmatrix}
            Z_P^\top & U_F^\top
        \end{bmatrix} \right)^{-1} \begin{bmatrix}
            Z_P \\ U_F
        \end{bmatrix},
    \end{equation*}
    with $E_F [Z_P^\top \; U_F^\top]$ being the sample covariance of the future innovation with past data and future inputs. Since, by assumption, the innovation is white and data are collected in open-loop, $E_F$ is uncorrelated to both past input/output samples and future input data, the sample covariance in turn converges in probability to zero. The result in \eqref{eq:asymptotic_predictor} ultimately follows from the linearity of the projection operator, according to which $\hat{E}_F$ also converges to zero in probability. $\hfill \blacksquare$
\end{pf}
Note that the data-based relationship in \eqref{eq:asymptotic_predictor} represents a common starting point for subspace identification for both LTI \cite{Favoreel:99a,VanOverschee:94a} and LPV systems \cite{Cox:21a}. At the same time, this asymptotically unbiased predictor still relies on the knowledge of the extended controllability matrix $\Gamma_z$ and of the extended Markov parameters matrix $\mathcal{H}_\mathrm{d}$.

\begin{table}[!tb]
    \centering
        \caption{Number of rows of the Hankel data matrices. Note that for an LPV system, $n_\textrm{q}=n_\textrm{p}+1 \geq 2$ and $n_\zeta \geq 3$ always hold.}
    \label{tab:Hankel_matrices_size}
    \begin{tabular}{|l|c|c|c|}
        \cline{2-4}
        \multicolumn{1}{c|}{} & $\bm{Z_P}$ & $\bm{U_F}$ & $\bm{Y_F}$ \\
        \hline
        n. rows & $(n_\textrm{u} n_\zeta + n_\textrm{y} n_\textrm{q}) \dfrac{n_\zeta^M-1}{n_\zeta-1} n_\textrm{q}^T$ & $n_\textrm{u} n_\textrm{q} \dfrac{n_\textrm{q}^T - 1}{n_\textrm{q} - 1}$ & $n_\textrm{y} T$\\
        \hline
    \end{tabular}
\end{table}

Nonetheless, given the linear structure of \eqref{eq:asymptotic_predictor}, the output predictor can be defined from data using Willems' fundamental lemma \cite{Willems:05a}. To this end, we use the data in $\mathfrak{D}_{N_{\textrm{data}}}$ to define the Hankel matrices $Z_P$, $U_F$, and $Y_F$ as in \eqref{eq:Hankel_matrices} with $k=M$. The dimensions of the matrices are specified in \tablename{~\ref{tab:Hankel_matrices_size}}. The data in $\mathfrak{D}_{N_{\textrm{data}}}$ are assumed to satisfy the following condition, which also characterize a key property of these Hankel matrices.
\begin{assum}[Persistence of excitation]\label{assum:persistence_excitation}
    The input and scheduling sequences comprised in the dataset $\mathfrak{D}_{N_{\textrm{data}}}$ are persistently exciting, i.e., 
    \begin{equation}\label{eq:xi}
        \Xi =
        \begin{bmatrix}
            Z_P \\ U_F
        \end{bmatrix} = 
        \begin{bmatrix}
            Z_{M,1,N}^P \\ U_{M,1,N}^F
        \end{bmatrix},
    \end{equation} 
    has full row rank. 
\end{assum}
\begin{rem}[Persistently exciting scheduling]
    In the LPV context, the rank condition in Assumption~\ref{assum:persistence_excitation} is influenced by the scheduling signal. When the scheduling signal $p$ is endogenous, such a requirement is generally satisfied during data collection. Indeed, persistently exciting inputs would lead to a persistently exciting scheduling signal. However, when $p$ is exogenous and not controllable, guaranteeing persistence of excitation might be nontrivial. Future work will be devoted to addressing this issue. 
\end{rem}

For a fixed scheduling trajectory, \eqref{eq:asymptotic_predictor} is linear in the lifted variables $Z_P$ and $U_F$ and can be interpreted as a relaxed LTI embedding of the LPV predictor, as discussed in \cite{Markovsky:26a}. Thus, under Assumption~\ref{assum:persistence_excitation}, we can formulate a data-driven predictor through Willems' Lemma~\cite{Willems:05a}. \\
To this end, for a single online trajectory with past window $[-M,-1]$ and prediction window $[0,T-1]$, we introduce the shorthand notation
\begin{equation} \label{eq:shorthand-notation}
    z^P := z^P_0,\quad
    u^F := u^F_0,\quad
    y^F := y_{[0,T-1]},
\end{equation}
where $z^P_0$ and $u^F_0$ are defined according to \eqref{eq:definition-of-zP} and \eqref{eq:definition-of-uF}.
\begin{lem}[Data-based LPV predictor] \label{lem:Willems_lemma}
    Under Assumption~\ref{assum:persistence_excitation} and for a fixed scheduling trajectory $p_{[-M, T-1]}$, a sequence $u_{[-M, T-1]}$, $y_{[-M, T-1]}$, is a feasible trajectory of the lifted predictor \eqref{eq:asymptotic_predictor} if and only if there exists $g \in \mathbb{R}^{N}$ such that
    \begin{equation}\label{eq:data-based_dynamics}
        \begin{bmatrix}
            z^P \\ u^F \\ y^F
        \end{bmatrix} = \begin{bmatrix}
            Z_P \\ U_F \\ Y_F
        \end{bmatrix} g.
    \end{equation}
\end{lem}
\begin{pf}
    The proof follows from that of the Willems' fundamental lemma in \cite{Berberich:20b}, with the difference that the LTI-like dynamics in \eqref{eq:asymptotic_predictor} are built from the scheduling-dependent past state $z^P$ and future input $u^F$. $\hfill \blacksquare$
\end{pf}

\begin{rem}[Scheduling dependence]
    In \eqref{eq:data-based_dynamics}, the input $u^F$ (see \eqref{eq:definition-of-uF}) depends on the \emph{future} scheduling trajectory $p_{[0,T-1]}$, whereas $y^F$ is scheduling-independent. Moreover, according to \eqref{eq:definition-of-zP}, $z^P$ depends on the entire scheduling trajectory $p_{[-M, T-1]}$.
\end{rem}

Lemma~\ref{lem:Willems_lemma} shows that, provided that the future scheduling signal $p_{[0,T-1]}$ is known, the future output $y^F$ can be obtained from a batch of data, for any given initial condition $z^P$ and future control action $u^F$. However, this predictor is not structurally consistent with the lifted LPV dynamics \cite{Markovsky:26a}. Indeed, even if every trajectory generated by the LPV system satisfies \eqref{eq:data-based_dynamics}, an arbitrary choice of $g$ yields a future input sequence $u^F = U_F g$ that does not necessarily respect the Kronecker structure in \eqref{eq:definition-of-uF}. Nevertheless, by introducing the block-diagonal matrix
    \begin{equation}\label{eq:large_calligraphic_P}
        \mathscr{P} = \begin{bmatrix}
    q_{T-1} \otimes \ldots \otimes q_0 \otimes I_{n_\mathrm{u}} & \cdots & 0\\
        \vdots & \ddots & \vdots  \\
        0 & \cdots & q_{T-1} \otimes I_{n_\mathrm{u}} \end{bmatrix},
    \end{equation}
we can show that the prediction error resulting from using \eqref{eq:data-based_dynamics} is bounded, as formalized next. 

\begin{thm}[Prediction error bound] \label{thm:consistency_error}    
    Fix a future scheduling trajectory $p_{[0,T-1]}$ and a lifted regressor $z^P$. Let $g\in\mathbb{R}^N$ satisfy $Z_P g=z^P$, and define the relaxed lifted input and the corresponding predicted output as $u^F = U_F g$ and $y^F = Y_F g$.
    Let $\tilde{u}^F$ be the input sequence closest to $u^F$ that respects the Kronecker structure in \eqref{eq:definition-of-uF}. Moreover, let $\tilde{y}^F$ be the output sequence generated by the predictor from $\tilde{u}^F$ and $z^{P}$. Then the following holds:
    \begin{subequations}
    \begin{align}
        \varepsilon_u & = \left\lVert (I - \mathscr{P}_\Pi) U_F g \right\rVert_2, \\
        \varepsilon_y & \leq \left\lVert \mathcal{H}_\mathrm{d} \right\rVert_2 \left\lVert (I - \mathscr{P}_\Pi) U_F g \right\rVert_2, 
    \end{align}
    where 
    \begin{equation}
        \varepsilon_u =\lVert u^F - \tilde{u}^F \rVert_2, \quad
        \varepsilon_y = \lVert y^F - \tilde{y}^F \rVert_2,
    \end{equation}
    $\mathscr{P}$ is defined in \eqref{eq:large_calligraphic_P} and 
    \begin{equation}\label{eq:projected_p}
    \mathscr{P}_\Pi = \mathscr{P} \left( \mathscr{P}^\top \mathscr{P} \right)^{-1} \mathscr{P}^\top.
    \end{equation} 
    \end{subequations}
\end{thm}
\begin{pf}
    Let $u_{[0, T-1]}$ be a scheduling-independent input sequence. From \eqref{eq:definition-of-uF}, any future input sequence consistent with the LPV dynamics can be written as $u^F =\mathscr{P} u_{[0,T-1]}$, with $\mathscr{P}$ defined in \eqref{eq:large_calligraphic_P}. 
    The consistency error $\varepsilon_u$ can thus be computed as the distance of $U_F g$ from its projection on $\mathrm{range}(\mathscr{P})$, i.e., 
    \begin{align*}
        \varepsilon_u\!\! & =\! \left \lVert u^F \!\!\!-\! \tilde{u}^F \right \rVert_2\!=\! \left \lVert U_F g \!-\! \Pi_\mathscr{P} \left(U_F g\right) \right \rVert_2 \!=\! \left \lVert (I \!\!-\! \mathscr{P}_\Pi) U_F g \right \rVert_2.
    \end{align*}
    Since $z^P$ is fixed, the predicted outputs associated with $u^F = U_F g$ and $\tilde{u}^F=\Pi_\mathscr{P}(U_F g)$ yield
    \begin{align*}
        \varepsilon_y & = \lVert y - \tilde{y} \rVert_2 \!=\! \left \lVert \Gamma_z z^P \!+\! \mathcal{H}_\mathrm{d} u^F \!-\! \Gamma_z z^P \!-\! \mathcal{H}_\mathrm{d} \tilde{u}^F \right \rVert_2 \\
        & = \left \lVert \mathcal{H}_\mathrm{d} \left( u^F - \tilde{u}^F \right) \right \rVert_2  \leq \left\lVert \mathcal{H}_\mathrm{d} \right\rVert_2 \left\lVert (I - \mathscr{P}_\Pi) U_F g \right\rVert_2, 
    \end{align*}
    which concludes the proof.
    $\hfill \blacksquare$
\end{pf}
Note that structural consistency could eventually be enforced by explicitly imposing $u^F = \mathscr{P}u_{[0,T-1]}$, where $u_{[0,T-1]}$ is the scheduling-independent part of $u^F$, as done in \cite{Verhoek:21a,Verhoek:25a,Verhoek:25b} with a different data-driven predictor. This parameterization would also reduce the dimension of the input decision variable in a control problem, as $u_{[0,T-1]}$ has fewer components than $u^F$. In our formulation, however, we keep the scheduling signal directly embedded in $u^F$. As shown later, this choice enables an efficient optimization reformulation, in which the number of online decision variables can be reduced compared with that obtained by explicitly enforcing $u^F = \mathscr{P}u_{[0,T-1]}$.

% However, while setting $u^F = \mathscr{P} u$ directly guarantees the structural consistency of the predictor, it also requires accurate knowledge of the future scheduling signal, which is often noisy and is only approximately available. Not enforcing this structural requirement could thus reduce prediction bias and improve robustness, at the expense of increased prediction variance.

\section{LPV-DeePC and its limitations}\label{sec:limitations}
Building on the data-driven predictor \eqref{eq:data-based_dynamics}, we can now formulate a finite-horizon optimal control problem in a similar way to DeePC (see \cite{Coulson:19a}), but with scheduling-dependent Hankel matrices.

Let $y^r$ be an output reference signal, $\mathbb{U} \subseteq \mathbb{R}^{n_\textrm{u}}$ and $\mathbb{Y} \subseteq \mathbb{R}^{n_\textrm{y}}$ be polytopic sets defining input and output constraints, respectively. Moreover, let $u_{[0, T-1]}$ denote the scheduling-independent input sequence extracted from $u^F$. The LPV-DeePC problem is
\begin{subequations}\label{eq:DeePC}
    \begin{align}
        \min_{g} & \quad \sum_{k=0}^{T-1} \lVert y_k^F - y_{k}^{r} \rVert_Q^2 + \lVert u_k \rVert_R^2 \label{eq:DeePC_cost} \\
        \text{s.t.} & \quad \begin{bmatrix}
            z^P \\ u^F \\ y^F
        \end{bmatrix} = \begin{bmatrix}
            Z_P \\ U_F \\ Y_F
        \end{bmatrix} g,  \label{eq:DeePC_predictor}\\
        & \quad u_k \in \mathbb{U}, \; y_k^F \in \mathbb{Y}
         \label{eq:DeePC_saturation}
    \end{align}
\end{subequations}
with $Q \in \mathbb{R}^{n_\textrm{y} \times n_\textrm{y}}$, $R \in \mathbb{R}^{n_\textrm{u} \times n_\textrm{u}}$ being positive definite matrices that weight the tracking error and the control effort, respectively. 
Although the predictor \eqref{eq:data-based_dynamics} depends on the lifted input $u^F$, the cost and constraints are imposed only on the scheduling-independent inputs $u_k$, since these are the physical control actions subject to actuation limits.

As remarked in the previous section, the predictor used in \eqref{eq:DeePC} is not structurally consistent with the lifted LPV dynamics. While structural consistency is desirable for accurate output prediction, ``small''  inconsistencies can be tolerated by MPC/DPC problems (see, e.g., the discussion in  \cite{Changrui:24a,Patwardhan:89a}). At the same time, the predictor's inconsistency can be circumvented by explicitly imposing $u^F = \mathscr{P}u_{[0,T-1]}$ or by changing the cost function in~\eqref{eq:DeePC_cost} as
    \begin{equation}
        \sum_{k=0}^{T-1} \lVert y_k^F \!-\! y_{k}^{r} \rVert_Q^2 + \lVert u_k \rVert_R^2 + \lambda_\varepsilon \left\lVert (I \!-\! \mathscr{P}_\Pi) U_F g \right\rVert_2^2,
    \end{equation}
where $\mathscr{P}_{\Pi}$ is defined in \eqref{eq:projected_p} and $\lambda_\varepsilon \geq 0$ is a tunable hyperparameter. Both options require fixing a future scheduling trajectory, either known a priori or iteratively estimated (see Remark~\ref{rem:scheduling_prediction}). 
Once such a trajectory is fixed, enforcing consistency introduces two intertwined problems of practical and computational nature. Indeed, guaranteeing Assumption~\ref{assum:persistence_excitation} implies imposing on the number of available data $N_{\textrm{data}}$ a lower bound that depends on the dimensions of inputs, outputs, and scheduling signal, as well as prediction horizon $T$ and the past horizon $M$. Hence, as shown in Table~\ref{tab:data_for_excitation}, the amount of data required to guarantee persistence of excitation increases exponentially and, even for relatively small values\footnote{$M$ should taken sufficiently large to guarantee that the initial condition $X_0$ of the collected trajectory can be neglected, as is well known in subspace identification~\cite{Cox:21a,vanWingerden:09a,Verdult:05a}).} of $M$ and $T$, the data collection becomes easily impractical. As the column size of~$\Xi$ in \eqref{eq:xi} must be greater than or equal to its exponentially increasing row size to guarantee persistence of exitation, then also the size of the optimization variable $g$ grows exponentially (a well-known problem in the literature \cite{Verheijen:23a}). Hence, enforcing consistency requires not only an impractical amount of data, but it also leads to a problem that is computationally intractable (especially when solved in real time).

\begin{table}[!tb]
    \centering
    \caption{Minimum number of data needed to have the Hankel matrix in \eqref{eq:data-based_dynamics} full row rank for the smallest possible LPV system, i.e., for $n_\textrm{p} = n_\textrm{u} = n_\textrm{y} = 1$.}
    \label{tab:data_for_excitation}
    \begin{tabular}{|c|cccc|}
        \cline{2-5}
        \multicolumn{1}{c|}{} & $\bm{T=1}$ & $\bm{T=5}$ & $\bm{T=10}$ & $\bm{T=20}$ \\
        \hline 
        $\bm{M=1}$ & 14 & 232 & 7186 & $7.3 \cdot 10^{6}$ \\
        \hline
        $\bm{M=5}$ & 1218 & 19436 & 619520 & $6.4 \cdot 10^{8}$ \\
        \hline
        $\bm{M=10}$ & 295253 & $4.7 \cdot 10^{6}$ & $1.5 \cdot 10^{8}$ & $1.5 \cdot 10^{11}$\\
        \hline
    \end{tabular}  
\end{table}

To cope with the increasing dimension of the optimization variable with the size of the dataset, \cite{Verhoek:25a} proposes to modify the multistep predictor into a recursive, one-step-ahead predictor. By doing so, a trade-off arises between computational complexity and suboptimality (see \cite{Kohler:22a} for a discussion), making the problem more suitable for online implementation. Rather than modifying the nature of the predictor, we instead take inspiration from the $\gamma$-DDPC framework proposed in \cite{Breschi:23a}. We thus compute the LQ decomposition 
\begin{equation}\label{eq:LQ_decomposition}
    \begin{bmatrix}
        Z_P \\ U_F \\ Y_F
    \end{bmatrix} = \begin{bmatrix}
        L_{11} & 0 & 0 \\
        L_{21} & L_{22} & 0 \\
        L_{31} & L_{32} & L_{33}
    \end{bmatrix} \begin{bmatrix}
        Q_1 \\ Q_2 \\ Q_3
    \end{bmatrix},
\end{equation}
where, under Assumption~\ref{assum:persistence_excitation}, $L_{i,i}$, $i \in \{1,2,3\}$, are non-singular and $Q_i$, $i \in \{ 1,2,3 \}$, are orthonormal matrices. Substituting \eqref{eq:LQ_decomposition} into \eqref{eq:data-based_dynamics}, then the following holds:
\begin{equation}
    \begin{bmatrix}
        z^P \\ u^F \\ y^F
    \end{bmatrix} = \begin{bmatrix}
        L_{11} & 0 & 0 \\
        L_{21} & L_{22} & 0 \\
        L_{31} & L_{32} & L_{33}
    \end{bmatrix} \begin{bmatrix}
        \gamma_1 \\ \gamma_2 \\ \gamma_3
    \end{bmatrix},
\end{equation}
where $z^P$ and $u^F$ satisfy the structure in \eqref{eq:single_predictor} and \eqref{eq:definition-of-uF}, respectively, and $\gamma_1 = Q_1 g$, $\gamma_2 = Q_2 g$, and $\gamma_3 = Q_3 g \in \mathbb{R}^{Tn_\textrm{y}}$. Since the scheduling signal $p_{[-M, -1]}$ is known at the time instant $k=0$, the initial condition $z^P$ is known and $\gamma_1$ can be explicitly solved a priori, i.e., 
\begin{equation}\label{eq:gamma_1_opt}
 \gamma_1^\star = L_{11}^{-1} z^P,   
\end{equation}
 allowing us to finally recast the LPV version of the $\gamma$-DDPC problem as
\begin{equation}\label{eq:gammaDDPC}
    \begin{aligned}
        \min_{\gamma_2, \gamma_3} & \sum_{k=0}^{T-1} \lVert y_k^F \!-\! y_{k}^r \rVert_Q^2 \!+\! \lVert u_k \rVert_R^2 \!+\! \beta_2 \lVert \gamma_2 \rVert_2^2 \!+\! \beta_3 \lVert \gamma_3 \lVert_2^2, \\
        \text{s.t.} & \quad \begin{bmatrix}
            u^F \\ y^F
        \end{bmatrix} \!\!=\! \begin{bmatrix}
            L_{21} & L_{22} & 0 \\
            L_{31} & L_{32} & L_{33}
        \end{bmatrix} \!\!\begin{bmatrix}
            \gamma_1^\star \\ \gamma_2 \\ \gamma_3
        \end{bmatrix}, \\
        & \quad \gamma_1^\star = L_{11}^{-1} z^P, \\
        & \quad u_k \in \mathbb{U}, \; y_k^F \in \mathbb{Y}, 
    \end{aligned}
\end{equation}
where $\gamma_2$ and $\gamma_3$ have a dimension that depends on the number of rows in $U_F$ and $Y_F$, respectively, while  $\beta_2, \beta_3 \geq 0$ weight the 2-norm regularizations on $\gamma_2$ and $\gamma_3$, used to empirically counteract the impact of inconsistencies and measurement noise (see, e.g., \cite{Breschi:23b}). This reformulation allows us to obtain a design problem with optimization variables whose dimension is independent of the number of data available to construct the data-driven predictor. However, the dimension of the optimization problem still depends on the row size of $Z_{P}$ and $U_F$, and, hence, it may become easily intractable.   

\begin{rem}[Knowledge of the scheduling signal]\label{rem:scheduling_prediction}
    So far, we have assumed that the scheduling trajectory over the prediction horizon is known. However, such an assumption does not often hold in practice. Indeed, in many applications, the scheduling signal is defined by an input/output dependent \emph{scheduling map} $\phi$. Assuming such map to be known a priori and given $p_k:=\phi(u_,y_k)$, a possible strategy to estimate future scheduling signals is the Gain-Scheduling (GS) approach \cite{Rugh:00}. In this case, $p_i$ is kept equal to $\phi(u_k,y_k)$ for all $i = k+1, \dots, k+t$, hence neglecting the parameters' variation induced by the scheduling in \eqref{eq:LPV-SS-original}. Alternatively, the future scheduling trajectory can be iteratively refined using the Sequential Scheduling Synthesis (SSS) method proposed in~\cite{cisneros2020nonlinear} and already applied to data-driven predictive control in~\cite{Verhoek:25a}. In this case, the predictive control problem is iteratively solved multiple times per time step. At every iteration, the optimal inputs and predicted outputs are used to update $p_{[k+1, k+t]}$ using the scheduling map~$\phi$. The updated scheduling trajectory will be used in the next iteration until convergence. While further increasing the computational complexity of the control scheme, this second approach is guaranteed, under mild conditions (see~~\cite{hespe2021convergence}) to have local contraction properties and, thus, convergence guarantees.   
\end{rem}

\section{Tackling LPV-DDPC intractability}\label{sec:intractability}
After performing the LQ decomposition of the Hankel data matrices to limit the DeePC's exponentially increasing number of optimization variables, we now handle the issues related to the dimensions of $Z_P$ and $U_F$ (see \tablename{~\ref{tab:Hankel_matrices_size}}) by building a reduced-order approximation of the data-based dynamics. To this end, we examine the information brought about by the rows of $Z_P$ and $U_F$. In this respect, we notice that, due to the Kronecker products, most rows of $Z_P$ and $U_F$ contain high-order elements in the scheduling variable, i.e., multiple products of its values at possibly different time instants. Therefore, whenever the scheduling signal has components with absolute magnitude smaller than one, then high-order elements would tend to zero and, as a consequence, most rows of $Z_P$ and $U_F$ would be uninformative. 

Although this condition is not generally satisfied by the controlled LPV-SS system \eqref{eq:LPV-SS-original}, the system can be equivalently written in a scaled form that satisfies such a condition, as the scheduling signal belongs to a bounded set $\mathbb{P}$. Specifically, we can normalize the scheduling variable's components as: 
\begin{equation}\label{eq:normalization_p}
    \tilde{p}_k^{[i]} = 2 \frac{p_k^{[i]} - \underline{p}^{[i]}}{\bar{p}^{[i]} - \underline{p}^{[i]}} - 1, \quad \tilde{p}_k^{[i]} \in [-1,1],
\end{equation} 
with $\underline{p}^{[i]}$ and $\bar{p}^{[i]}$ being the lower and upper bounds of the $i$-th component $p^{[i]}$, for $i=1=1,\ldots,n_\textrm{p}$. 
\begin{rem}[Normalization range]
    The normalization in~\eqref{eq:normalization_p} implies the knowledge of the lower and upper bounds of the scheduling signal. While assuming the availability of such bounds is common in the LPV framework, they can also be estimated from available data. However, using estimated bounds might lead to an ineffective normalization procedure, in turn causing the elimination of rows in $Z_P$ and $U_F$ that might be relevant to predict the system's output. At the same time, as discussed in \cite{Verdult:02a}, neglecting multiplications of the scheduling variables ultimately implies that one neglects high-order dynamics in $p$, which a limited number of data can in any case describe with limited accuracy.
\end{rem}

\begin{table}[!tb]
    \caption{Total number of rows of $Z_P$ when $n_\textrm{p} = n_\textrm{u} = n_\textrm{y} = 1$ without row reduction (Full) and with the reduced number of rows obtained by setting $h_Z = 3$.}
    \label{tab:a_priori_reduction}
    \centering
    \begin{tabular}{|l|c|c|c|c|}
        \cline{2-5}
         \multicolumn{1}{c|}{} & \diagbox{$\bm{M}$}{$\bm{T}$}& $\bm{5}$ & $\bm{10}$ & $\bm{20}$\\
         \hline 
        Full & \multirow{2}{*}{$\bm{5}$} & $1.9 \cdot 10^{4}$ & $6.2 \cdot 10^{5}$ &  $6.3 \cdot 10^{8}$\\
        \cline{1-1}\cline{3-5}
        Reduced & & $4.1 \cdot 10^{2}$ & $9.6 \cdot 10^{2}$ & $2.8 \cdot 10^{3}$\\
        \hline
        Full & \multirow{2}{*}{$\bm{10}$} & $4.7 \cdot 10^{6}$ & $1.5 \cdot 10^{8}$ &  $1.5 \cdot 10^{11}$\\
        \cline{1-1}\cline{3-5}
        Reduced & & $1.4 \cdot 10^{3}$ & $2.8 \cdot 10^{3}$ & $7.0 \cdot 10^{3}$\\
        \hline
    \end{tabular}
    %\begin{tabular}{c|ccc}
    %    & $\bm{T=5}$ & $\bm{T=10}$ & $\bm{T=20}$ \\
    %    \hline 
    %    \multirow{2}{*}{$\bm{M=5}$} & 19360 & 619520 & $6.3 \cdot 10^{8}$ \\
    %    & (405) & (955) & (2805) \\ \hline
    %    \multirow{2}{*}{$\bm{M=10}$} & $4.7 \cdot 10^{6}$ & $1.5 \cdot 10^{8}$ & $1.5 \cdot 10^{11}$ \\
   %     & (1410) & (2760) & (6960) 
   % \end{tabular}
\end{table}

As the normalization leads to an LPV-SS equivalent to the original one, we can use the dataset with normalized scheduling, i.e.,
\begin{equation}
    \tilde{\mathfrak{D}}_{N_{\textrm{data}}} = \{u_k, y_k, \tilde{p}_k \}_{k=0}^{N_\textrm{data}-1},
\end{equation}
to construct the data-driven predictor with reduced rows. Thanks to the normalization of the scheduling, a reasonable approximation of the system dynamics thus consists of removing the rows from $Z_P$ and $U_F$ that contain more than a given number of multiplications among the scheduling signal's components. This procedure is guided by the (user-defined) minimum number of products $h_Z \in \mathbb{N}$ ($h_U \in \mathbb{N}$) between scheduling variables required for removing a row of $Z_P$ ($U_F$). Note that choosing small $h_Z$, $h_U$ trades off accuracy for computational tractability, up to the extreme case $h_Z = h_U = 1$ (i.e., no row containing the scheduling variable), where the obtained reduced-order predictor becomes LTI. This trade-off is exemplified in \tablename{~\ref{tab:a_priori_reduction}}, highlighting how this simple reduction strategy allows one to use larger past and future horizons, which is generally not possible with other data-driven predictive LPV control techniques (see, e.g., the discussion in \cite{Verhoek:25a}).

\subsection{Keeping \textquotedblleft relevant\textquotedblright \ rows only}
Since neglecting high-order scheduling dynamics leads to Hankel matrices that may still contain thousands of rows (see \tablename{~\ref{tab:a_priori_reduction}}), we propose to further reduce $Z_P$ and $U_F$ by discarding additional rows based on their \textquotedblleft relevance\textquotedblright \ for describing $Y_F$. To define relevance, we first characterize the approximation error in reconstructing $Y_F$ achieved when using only \textquotedblleft relevant\textquotedblright \ rows. 

Let us reorder the rows of the Hankel matrix obtained stacking $
Z_P$ and $U_F$ to separate $\Omega_\textrm{rel}$, containing the most relevant rows of $
Z_P$ and $U_F$ for the reconstruction of $Y_F$, and $\Omega_\textrm{irrel}$, which comprises all the remaining ones. We can then perform an LQ decomposition of the resulting matrix, i.e.,
\begin{equation} \label{eq:full_LQ}
    \begin{bmatrix}
        \Omega_\textrm{rel} \\ \Omega_\textrm{irrel} \\ Y_F
    \end{bmatrix} = \begin{bmatrix}
        L_{11}^{\Omega} & 0 & 0 \\
        L_{21}^{\Omega} & L_{22}^{\Omega} & 0 \\
        L_{31}^{\Omega} & L_{32}^{\Omega} & L_{33}^{\Omega}
    \end{bmatrix} \begin{bmatrix}
        Q_1^{\Omega} \\ Q_2^{\Omega} \\ Q_3^{\Omega}
    \end{bmatrix},
\end{equation}
which is instrumental to formalize the approximation error resulting from discarding all the non-relevant rows $\Omega_\mathrm{irrel}$ as follows.
\begin{lem}[Approximation error by row reduction]
    The mismatch between $Y_F$ and $Y_{F, \mathrm{rel}}=L_{31}^\Omega Q_{1}^\Omega+L_{33}^\Omega Q_{3}^\Omega$ is 
    \begin{equation}
        \lVert Y_{F} - Y_{F, \mathrm{rel}} \rVert_F = \lVert L_{32}^\Omega \rVert_F. 
    \end{equation}
\end{lem}
\begin{pf}
    The proof follows from the fact that $Q_{2}^\Omega$ is orthonormal and $Y_F=L_{31}^\Omega Q_{1}^\Omega+L_{32}^\Omega Q_{2}^\Omega+L_{33}^\Omega Q_{3}^\Omega$. $\hfill \blacksquare$
\end{pf}
According to this result, we quantify the relevance, for example, of the $k$-th row of $Z_P$, namely $Z_P[k,:]$, by checking the residual 
\begin{equation}\label{eq:residual_equation}
    r_k = \min_\psi \left \lVert Y_F - \psi Z_P[k,:] \right \rVert_F^2.
\end{equation}
Since rows with small $r_k$ contribute more to $Y_F$, we can then use this criterion for \emph{iterative}, relevance-based row removal, as summarized\footnote{The code for the algorithms to reproduce our results is publicly available at \url{https://github.com/fepor99/LPV-gamma-DDPC}. } in Algorithms~\ref{alg:residuals},~\ref{alg:row_selection}. Algorithm~\ref{alg:residuals} is the outer reduction routine. It first forms the candidate sets of rows of $Z_P$ and $U_F$ by discarding high-order scheduling terms according to $h_Z$ and $h_U$. Then, it selects $n_{Z_P}$ rows from $Z_P$ and $n_{U_F}$ rows from $U_F$, while constructing the reduced LQ factor $L^{\rm red}$. Note that $n_{Z_P}$ and $n_{U_F}$ are user-chosen hyperparameters. At each row selection step, Algorithm~\ref{alg:residuals} calls the subroutine \textsc{SelR} in Algorithm~\ref{alg:row_selection}. This subroutine computes the residuals associated with the current candidate rows (line~\ref{alg:row_selection:recursive} in Algorithm~\ref{alg:row_selection}), selects the row $k^*$ yielding the smallest residual, and returns the corresponding Householder transformation $H_k^\star$ used to remove its contribution from the remaining data (line~\ref{alg:residuals:householder} of Algorithm~\ref{alg:row_selection}). Thus, the procedure alternates between selecting the most relevant row, projecting the data onto the orthogonal complement of the selected row, and updating the LQ factorization. 
Note that, before the relevance analysis, we select all rows of $Z_P$ and $U_F$ that are not multiplied by the scheduling signal (see line~\ref{alg:residuals:indipendent} of Algorithm~\ref{alg:row_selection}). This initial step forces the predictor to always contain the LTI portion of the dynamics.

\begin{rem}[Selecting $n_{Z_P}$ and $n_{U_F}$]\label{rem:SVD}
    A heuristic to select $n_{Z_P}$ and $n_{U_F}$ consists of analyzing the residuals (e.g., \eqref{eq:residual_equation} for $Z_P$) of Algorithm~\ref{alg:residuals}. As progressively more rows of $Z_P$ and $U_F$ are selected and sorted by relevance, one can inspect the residual decrease as a function of the number of selected rows. Similar to other model-reduction approaches, e.g., SVD~\cite{VanOverschee:94a,VanOverschee:96a}, one can then determine by inspection a cut-off value for $n_{Z_P}$ and $n_{U_F}$, trading off loss of accuracy and computational complexity.
\end{rem}

\algrenewcommand\algorithmicindent{0.3em}
\begin{algorithm}[t]
    \caption{Reduced-order LQ decomposition}\label{alg:residuals}
    \begin{algorithmic}[1]
        % Inputs
        \Require Dataset $\mathfrak{D}_{N_{\textrm{data}}}$; past horizon $M$; future horizon $T$; desired number of rows $n_{Z_P}$ and $n_{U_F}$; threshold on high-order elements $h_Z$ and $h_U$.
        
        % Outputs
        \Ensure Reduced-order LQ decomposition matrix $L^{\mathrm{red}}$; sets of selected row indices $\mathbb{I}_Z$ and $\mathbb{I}_U$.\vskip 3pt 
        \hrule \vskip 3pt
        
        % Start code
        \Statex \textbf{Initialization}
        \State $Y \gets Y_F$, \; $\mathbb{I}_Z \gets \emptyset$, \; $\mathbb{I}_U \gets \emptyset$, \; $L^{\mathrm{red}} \gets 0$; 
        \State $\varrho_0=0$, \; $\mathcal{H} \gets I$;
        \State \textbf{compute} the set $\mathbb{H}_Z$ of row indices of $Z_P$  with less than $h_Z$-order terms in $p$; 
        \State \textbf{compute} the set $\mathbb{H}_U$ of row indices of $U_F$ with less than $h_U$-order terms in $p$;

        % Select from ZP
        \Statex \textbf{Order reduction and LQ decomposition}
        \For {$j = 1, \ldots, n_{Z_P}$} \Comment{Choose from $Z_P$}
            \State $\{k^\star, K, \varrho_j,H_j^\star\} \gets \mathrm{\textsc{SelR}}(\mathbb{H}_Z,\mathcal{H},Z_P,Y,j,\varrho_{j-1})$;
            \State Remove $k^\star$ from $\mathbb{H}_Z$ and add it to $\mathbb{I}_Z$; 
            % $\mathbb{H}_Z \gets \mathbb{H}_Z \setminus \{k^\star\}$; $\mathbb{I}_Z \gets \mathbb{I}_Z \cup \{k^\star\}$;
            \State $Y \!\gets\! Y H_j^\star$, \; $\mathcal{H} \!\gets\! \mathcal{H} H_j^\star$, \; $L^{\mathrm{red}}[j,:] \!\gets\! K H_j^\star$;
        \EndFor

        % Select from UF
        \For {$i = 1, \ldots, n_{U_F}$} \Comment{Choose from $U_F$}
            \State $j \gets n_{Z_P} + i$; 
            \State $\{k^\star, K, \varrho_j,H_j^\star\} \gets $ \textsc{SelR}($\mathbb{H}_U, \mathcal{H}, U_F, Y, j$, $\varrho_{j-1}$);
            \State Remove $k^\star$ from $\mathbb{H}_U$ and add it to $\mathbb{I}_U$; 
            % \State $\mathbb{H}_U \gets \mathbb{H}_U \setminus \{k^\star\}$; $\mathbb{I}_U \gets \mathbb{I}_U \cup \{k^\star\}$;
            \State $Y \!\gets\! Y H_j^\star$, \; $\mathcal{H} \!\gets\! \mathcal{H} H_j^\star$, \; $L^{\mathrm{red}}[j,:] \!\gets\! K H_j^\star$;
        \EndFor

        % LQ decomposition of YF
        \For {$\ell = 1, \ldots, T n_\mathrm{y}$} \Comment{LQ decomposition of $Y_F$}
            \State $j \gets n_{Z_P} + n_{U_F} + \ell$;
            \State $K \gets\! Y[\ell, :]$;
            \State \textbf{find} the Householder matrix $H_j^\star$ and $S\in\mathbb{R}^{1\times j}$ such that $K H_j^\star = \begin{bsmallmatrix}
                S & 0
            \end{bsmallmatrix}$;
            \State $Y \!\gets\! Y H_j^\star$, \; $\mathcal{H} \!\gets\! \mathcal{H} H_j^\star$, \; $L^{\mathrm{red}}[j,:] \!\gets\! K H_j^\star$;
        \EndFor
    \end{algorithmic}
\end{algorithm}

\begin{algorithm}[t] 
    \caption{\textsc{SelR}($\mathbb{H}_X, \mathcal{H}, X, Y, j, \varrho_{j-1}$)} \label{alg:row_selection}
    \begin{algorithmic}[1]
    % Inputs
        \Require Candidate row set $\mathbb{H}_X$; projection matrix $\mathcal{H}$; data matrix $X$; output matrix $Y$; iteration $j$; previous residual $\varrho_{j-1}$.
        % Outputs
        \Ensure Selected row index $k^\star$; projected selected row $K$; updated residual $\varrho_j$; Householder matrix $H_{j^\star}$. \vskip 3pt 
        \hrule \vskip 3pt
        \State $\mathbb{H}_X^0 \leftarrow \{k\in \mathbb{H}_X : X[k,:]$ is independent of $p\}$;
        \If {$\mathbb{H}_X^0\neq\emptyset$} \label{alg:residuals:indipendent}
            \State $\mathbb{K}_X \leftarrow \mathbb{H}_X^0$; \Comment{select from rows independent of $p$} 
        \Else 
            \State $\mathbb{K}_X \leftarrow \mathbb{H}_X$;
        \EndIf
        \For {$k \in \mathbb{K}_X$}
            \State $\mathcal{X}_{k} \gets X[k,:] \mathcal{H}$;
            \State \textbf{find} the Householder matrix $H_k$ and $S\in\mathbb{R}^{1\times j}$ such that $\mathcal{X}_{k} H_k = \begin{bsmallmatrix}
                S & 0
            \end{bsmallmatrix}$; \label{alg:residuals:householder}
            \State $r_{k} \gets \begin{cases}
                \lVert Y[:,j\!:\!N] \, H_k[:, 2\!:\!N\!-\!j\!+\!1] \rVert_F^2 & \textrm{if } j=1; \\
                \varrho_{j-1} - \lVert Y[:,j\!:\!N] \, H_k[:,1] \rVert_2^2 & \textrm{if } j>1;
            \end{cases}$ \label{alg:row_selection:recursive}
        \EndFor
        \State $k^\star \gets \arg \min_k r_{k}$;
        \State $\varrho_{j} \gets \min_k r_{k}$;
        \State $K \gets \mathcal{X}_{k^\star}$;
        \State $H_j^\star \gets H_{k^\star}$; 
    \end{algorithmic}
\end{algorithm}

\section{Reduced-complexity LPV $\gamma$-DDPC and its comparison with existing approaches}\label{sec:reduced_gamma}
By leveraging the indices of the selected rows of $Z_P$ and $U_F$, described by the sets $\mathbb{I}_Z$ and $\mathbb{I}_U$, determined by Algorithm~\ref{alg:residuals}, we can now build the reduced-complexity version of $\gamma$-DDPC for LPV systems. Specifically, let $\bar{z}^P$ denote the subvector of $z^P$ in \eqref{eq:single_predictor} containing the entries whose indices belong to $\mathbb{I}_Z$. Using $L^{\mathrm{red}}$ matrix returned by Algorithm~\ref{alg:residuals}, we can then find $\gamma_1^\star$ as
\begin{equation}\label{eq:gamma_1_star_red}
    \gamma_1^\star = (L_{11}^{\mathrm{red}})^{-1} \bar z^P \in \mathbb{R}^{n_{Z_{P}}}.
\end{equation}
Accordingly, by denoting the subvector of $u^F$ in \eqref{eq:future_prediction} containing its entries with indices in $\mathbb{I}_U$ as $\bar{u}^F$, we can formulate the reduced-complexity $\gamma$-DDPC problem as:
\begin{equation}\label{eq:reduced_gammaDDPC}
    \begin{aligned}
        \min_{\gamma_2, \gamma_3} &\sum_{k=0}^{T-1} \lVert y_k^F \!-\! y_{k}^{r} \rVert_Q^2 \!+\! \lVert u_k \rVert_R^2 \!+\! \beta_2 \lVert \gamma_2 \rVert_2^2 \!+\! \beta_3 \lVert \gamma_3 \rVert_2^2, \\
        \text{s.t.} & \quad \begin{bmatrix}
            \bar u^F \\ y^F
        \end{bmatrix} \!\!=\!\! \begin{bmatrix}
            L_{21}^{\mathrm{red}} & L_{22}^{\mathrm{red}} & 0 \\
            L_{31}^{\mathrm{red}} & L_{32}^{\mathrm{red}} & L_{33}^{\mathrm{red}}
        \end{bmatrix} \begin{bmatrix}
            \gamma_1^\star \\ \gamma_2 \\ \gamma_3
        \end{bmatrix}, \\
        & \quad u_k \in \mathbb{U}, \; y_k^F \in \mathbb{Y}. 
    \end{aligned}
\end{equation}
Note that the optimization variable $\gamma_2 \in \mathbb{R}^{n_{U_F}}$, as well as $\gamma_1^\star$ have now a size dependent on the user-defined choices of $n_{Z_{P}}$ and $n_{U_{F}}$, whose tuning is key (as already remarked) to obtain a predictive control problem that is computationally tractable for real-time implementation.

\begin{table*}[!tb]
    \caption{Qualitative comparison of reduced-complexity $\gamma$-DDPC with existing LPV data-driven predictive control approaches. We evaluate whether they can handle noise, exploit consistent predictors, and assess their online computational complexity. IO stands for Input-Output, while CE denotes the use of the Certainty Equivalence principle~\cite{Astrom:95a,Hjalmarsson:05a}.}
    \label{tab:compairson_algos}
    \centering
    \begin{tabular}{|l|c|c|c|c|}
        \cline{2-5}
         \multicolumn{1}{c|}{} & \textbf{Noise} & \textbf{Consistency} & \textbf{Requirements} & \textbf{Online Complexity} \\
         \hline
         Polytopic DeePC~\cite{BouHamdan:24a} & $\textcolor{red!70!black}{\bm{\times}}$ & $\textcolor{red!70!black}{\bm{\times}}$ & Scheduling signal's controllability & \textcolor{orange!90!black}{Moderate}\\ 
         \hline
         LPV-IO-DPC~\cite{Verhoek:25a} & $\textcolor{green!70!black}{\bm{\checkmark}}$ & $\textcolor{orange!90!black}{\bm{\sim}}$ (slacks) & Existence of shifted-affine IO realization & \textcolor{red!70!black}{High} \\ 
         \hline
         SPC LPV~\cite{Dong:09a} & $\textcolor{orange!90!black}{\bm{\sim}}$ (CE) & $\textcolor{green!70!black}{\bm{\checkmark}}$ & Existence of affine SS realization & \textcolor{green!70!black}{Low}\\ 
         \hline
          Reduced $\gamma$-DDPC \textbf{(ours)} & $\textcolor{green!70!black}{\bm{\checkmark}}$ & $\textcolor{red!70!black}{\bm{\times}}$ & Existence of affine SS realization & \textcolor{orange!90!black}{Reduced}\\ 
         \hline
    \end{tabular}
\end{table*}

As summarized in \tablename{~\ref{tab:compairson_algos}}, this computational efficiency is one of the advantages of the proposed scheme over existing approaches for data-driven LPV predictive control, together with the relaxed requirements on the controlled system and the capability of our approach to cope with noisy setting. The latter is indeed one of the main distinctive features of our method with respect to the polytopic DeePC in \cite{BouHamdan:24a}. Such an approach relies on the assumption that the scheduling signal belongs to a known scheduling region $\mathbb{P}$ and, unlike ours, considers a noise-free setting. This hypothesis allows for the representation of the LPV predictor as a convex combination of as many frozen LTI predictors as the vertices of $\mathbb{P}$. Accordingly, polytopic DeePC does not hinge on the prediction of the future scheduling trajectory, which is advantageous when the scheduling signal is exogenous or unmeasurable, yet conservative whenever the future scheduling trajectory is not constant or an approximate prediction of the scheduling signal is available. Moreover, unlike us, \cite{BouHamdan:24a} requires collecting a distinct dataset for each vertex of the scheduling polytope $\mathbb{P}$. This implies that the scheduling signal must be controllable and held constant during each data-collection experiment. Such an assumption can be restrictive in many practical applications and is likely infeasible when the scheduling signal depends on the system's inputs and/or outputs. This requirement does not characterize the extensions of data-enabled predictive control (DeePC)~\cite{Coulson:19a} to LPV systems, including the LPV-IO-DPC formulation in \cite{Verhoek:25a}. Nonetheless, these approaches require the LPV system \eqref{eq:LPV-SS-original} to admit a shifted-affine LPV input–output realization, which is more restrictive compared to our (non-minimal) affine LPV-SS assumption. Moreover, they are conceived for a noise-free setting. The latter limitation is overcome by the heuristic introduction of slack variables to handle noise by softening the control problem's constraints. 

Meanwhile, both LPV-IO-DPC and LPV $\gamma$-DDPC rely on the knowledge of the future scheduling trajectory, yet such information is used in a completely different way within the two approaches. In particular, while in LPV-IO-DPC the predicted scheduling is explicitly embedded in the predictor by enforcing structural consistency constraints, we use the predicted scheduling only to define $\bar{z}^P$ in \eqref{eq:gamma_1_star_red}. As a consequence, our reduced approach does not guarantee structural consistency of the predicted inputs with \eqref{eq:definition-of-uF}, with the benefit of (\emph{i}) not enforcing potentially inaccurate scheduling information when the scheduling is inferred from data, and (\emph{ii}) leading to a scheme that is computationally more advantageous than LPV-IO-DPC.  

Our approach can also be connected to subspace predictive control and, specifically, to its LPV version (SPC LPV) introduced in \cite{Dong:09a}. As for LPV $\gamma$-DDPC, SPC LPV relies on the steps carried out in Section~\ref{sec:data_to_predictions} to obtain the predictor, yet estimating the matrices $\Gamma_z$ and $\mathcal{H}_\textrm{d}$ in \eqref{eq:asymptotic_predictor} and using them to have a certainty equivalent predictor, rather than exploting a data-driven predictor. This choice allows SPC LPV to use a predictor coherent with the Kronecker structure in \eqref{eq:definition-of-uF}, but neglects potential predictions errors, which LPV $\gamma$-DDPC can instead account for through $\gamma_3$ as for the LTI case (see \cite{Breschi:23b,Mattsson:24a}). 

\section{Sensitivity analysis: a numerical study} \label{sec:hyperparameter_analysis}
LPV $\gamma$-DDPC requires the selection of several hyperparameters, among which the number of rows $n_{Z_P}$ and $n_{U_F}$ dictating the complexity of the predictor, as well as the minimum number of products $h_Z$ and $h_U$ between scheduling signal at possibly different time instants to avoid constructing superfluous rows in the Hankel data matrices. We now analyze the sensitivity of LPV $\gamma$-DDPC to these parameters, providing a set of practical guidelines for their selection. To this end, we consider the same $4$-th order LPV system with $2$ inputs, $3$ outputs, and $3$-dimensional exogenous scheduling considered in \cite{Verdult:02a}. In all tests, we set ourselves in the same setting used therein (i.e., the innovation is a zero-mean, white noise sequence with a Gaussian distribution and $\sigma_e^2=1$), using exactly the same data collection procedure established therein to gather $N_\textrm{data}=2000$ input/output/scheduling samples starting from zero initial conditions. Our sensitivity analysis is carried out by closing the loop with LPV $\gamma$-DDPC, imposing the control inputs to be bounded to the set $\mathbb{U} = [-20,20] \times [-20,20]$, and steering the first and second outputs to track step reference signals. To do so, we set $Q = \textrm{diag}(10,1,0)$, $R = 0.01I$, $\beta_2 = 0.1$, and $\beta_3 = 0$ in the cost of \eqref{eq:reduced_gammaDDPC}, while we estimate the future scheduling using the GS approach. Note that the future scheduling already lies in $[-1,1]$ and is thus not normalized.

In our analysis, performance for varying values of the parameters of interest are quantitatively evaluated by looking at 
\begin{equation}\label{eq:indexes}
        J_u = \sum_{k=0}^{T-1} \lVert u_k \rVert_2^2, \quad
        J_y = \sum_{k=0}^{T-1} \left\lVert \frac{y_k - y_{k}^{r}}{y_{k}^{r}} \right\rVert_2^2, 
\end{equation}
evaluating the control effort and the tracking performance, respectively. The data collection, controller synthesis, and performance evaluation are performed over 30 Monte Carlo (MC) runs.

%\begin{equation}
\begin{figure*}[!tb]
    \centering
    \includegraphics[width=0.9\linewidth]{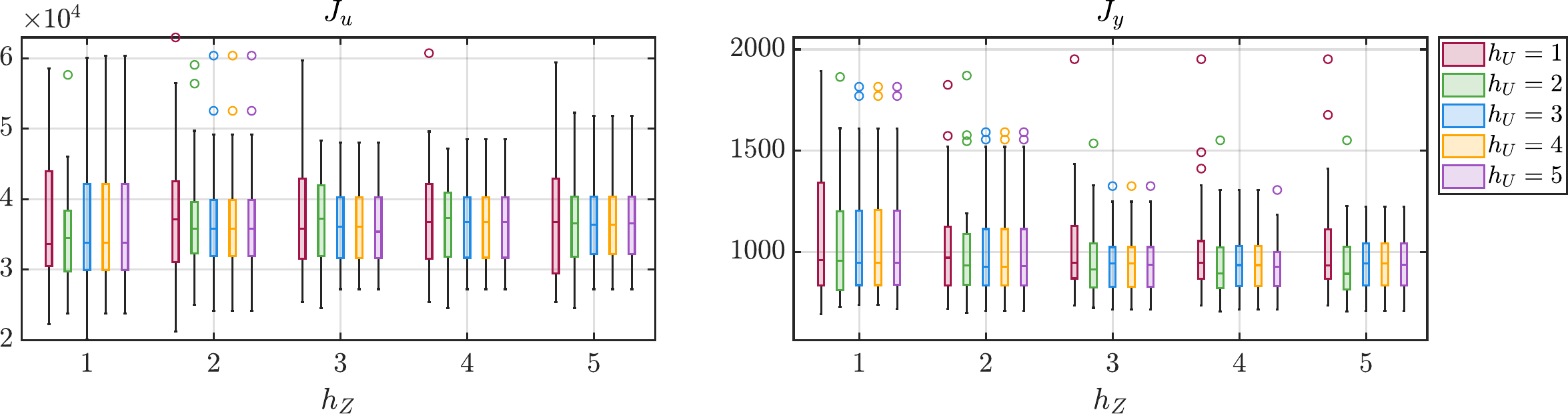}
    \caption{Performance metrics in \eqref{eq:indexes} over 30 MC runs for different values of $h_Z$ and $h_U$, with $n_{Z_P} = 30$ and $n_{U_F} = 20$.}
    \label{fig:sensitivity_hZ_hU}
\end{figure*}
\subsection{Sensitivity to higher-order dependency exclusion}
We start by fixing $n_{Z_P} = 30$ and $n_{U_F} = 20$ to analyze the sensitivity with respect to $h_{Z}, h_{U} \in \{1,2,3,4,5\}$. Note that, for higher values of $h_Z$ ($h_U$), the rows remaining after the reduction might be less than $n_{Z_P}$ ($n_{U_F}$). As shown in \figurename{~\ref{fig:sensitivity_hZ_hU}} (right panel), tracking performance is rather sensitive to the choice of $h_{Z}$  and $h_U$ when they are both low, achieving the worse value when the LPV dynamics are not considered, i.e., $h_Z = h_U = 1$. Nonetheless, it reaches a plateau around $h_Z = h_U = 3$, after which tracking performance remains almost invariant irrespective of the choice of these parameters. Similar conclusions can be drawn for $J_u$ (see the left panel of \figurename{~\ref{fig:sensitivity_hZ_hU}}), with the exception that input effort seems to be worse for $h_U=1$ and $h_Z=3$ or $h_Z=5$. Nonetheless, the plateau achieved in both cases for higher values of $h_Z$ and $h_U$ support our intuition that rows containing several products between the scheduling variable at (possibly) different time instants could be removed and still achieve satisfactory control performance in terms of control effort and tracking error.\\  
Concurrently, as shown in \tablename{~\ref{tab:synthesis_time}}, the time required to select the most relevant rows increases with both $h_Z$ and $h_U$. Hence, sticking to small values of $h_Z$ and $h_U$ makes the row selection more computationally efficient while not excessively deteriorating performance. 

\begin{table}[!tb]
    \caption{Average computational time (in [s]) to select the most relevant rows for different values of $h_Z$ and $h_U$ over $30$ Monte Carlo runs, with $n_{Z_P} = 30$ and $n_{U_F} = 20$.} \label{tab:synthesis_time}
    \centering
    \begin{tabular}{|c|c|c|c|c|c|}
        \hline 
        \diagbox{$\bm{h_U}$}{$\bm{h_Z}$}& $\bm{1}$ & $\bm{2}$ & $\bm{3}$ & $\bm{4}$ & $\bm{5}$ \\
        \hline 
        $\bm{1}$ & 2.03 & 3.09 & 6.62 & 30.43 & 234.59 \\
        \hline
        $\bm{2}$ & 2.44 & 3.92 & 7.37 & 31.06 & 235.28 \\
        \hline
        $\bm{3}$ & 2.91 & 4.83 & 8.24 & 32.14 & 236.31 \\
        \hline
        $\bm{4}$ & 4.19 & 6.05 & 9.84 & 34.07 & 237.94 \\
        \hline
        $\bm{5}$ & 6.04 & 8.05 & 11.78 & 38.23 & 240.80\\
        \hline
    \end{tabular}
\end{table}

\subsection{Sensitivity to number of rows}
By relying on the previous results, we now set $h_Z = h_U = 3$ and study the sensitivity of the achieved performance to the number of rows $n_{Z_P} \in \{15, 20, 25, 50, 100, 150 \}$ and $n_{U_F} \in \{ 10, 20, 30, 50 \}$ used to build the predictor. As highlighted in \figurename{~\ref{fig:sensitivity_n_ZP}}, when the number of selected rows is too small, e.g., $n_{Z_P}=15$ or $n_{U_F} = 10$, the predictor contains too little (if any) information about the LPV structure of the controlled system, leading to a deterioration in tracking performance (right panel) along with a general increase in control effort (left panel). At the same time, when an excessive number of rows is maintained, e.g., $n_{Z_P}=150$, the predictor starts fitting noise dynamics, and the reference signal is not tracked anymore. This result highlights a trade-off that should be achieved in tuning these parameters, with $n_{Z_P}$ and $n_{U_F}$ that should be large enough for the predictor to characterize all relevant plant dynamics, but not too large to avoid fitting the noise. \\
Following Remark~\ref{rem:SVD}, we then propose a heuristic to determine a suitable cut-off value for these two critical parameters by looking at the residual error in \eqref{eq:residual_equation}. Specifically, we first select large values for both $n_{Z_P}$ and $n_{U_F}$, plot the decrease in the prediction residuals as we select the relevant rows, and later cut when we recognize a knee in the residual reduction. For our example, we obtain the result shown in \figurename{~\ref{fig:prediction_residuals}}, from which we recognize that the residuals decrease shows a knee after 25 rows are taken from $Z_P$, and after 20 rows are taken from $U_F$. Therefore, we expect to be able to achieve satisfactory tracking performance and control efforts by setting $n_{Z_P} = 25$ and $n_{U_F} = 20$. This is indeed confirmed by \figurename{~\ref{fig:sensitivity_n_ZP}}.\\
It is worth pointing out that this heuristic for the selection of $n_{Z_P}$ and $n_{U_F}$ can be applied before closing the loop, as it relies on residuals that can be computed offline solely based on the available batch data. 

\begin{figure*}[!tb]
    \centering
    \includegraphics[width=0.9\linewidth]{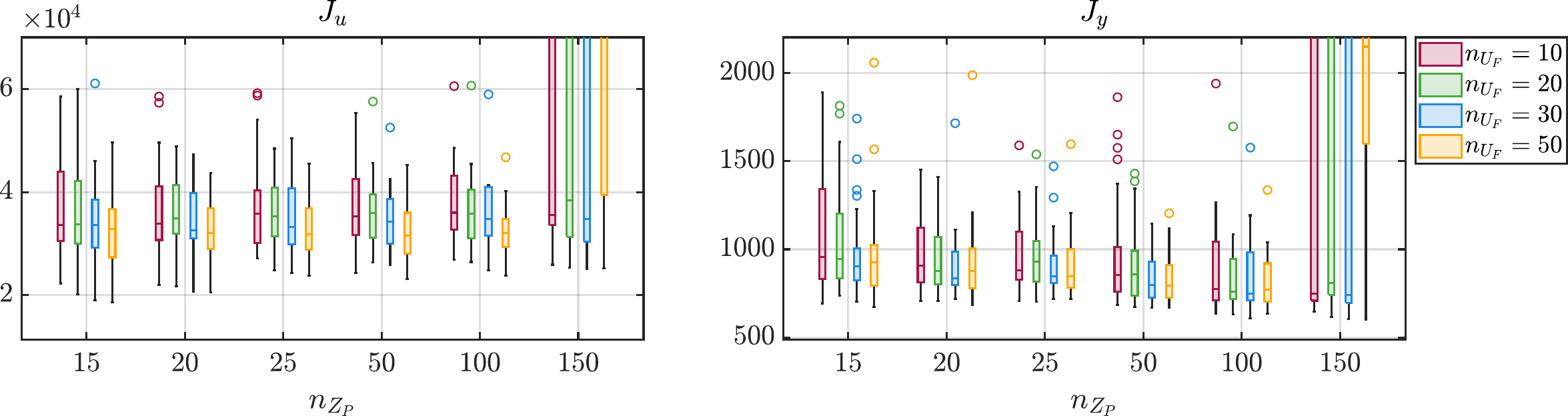}
    \caption{Performance metrics in \eqref{eq:indexes} over 30 MC runs for different values of $n_{Z_P}$ and $n_{U_F}$, for $h_Z = h_U = 3$.}
    \label{fig:sensitivity_n_ZP}
\end{figure*}

\begin{figure}[t]
    \centering
    \includegraphics[width=\linewidth]{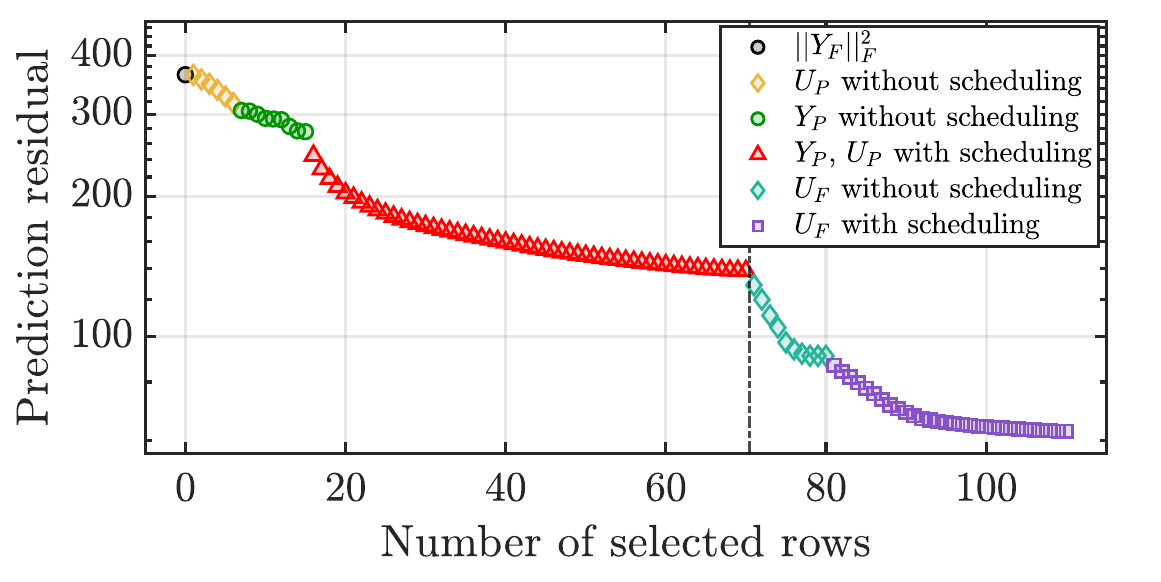}
    \caption{Prediction residuals \emph{vs} the number of selected rows. The vertical line divides the selection of the rows of $Z_P$ from those of $U_F$.}
    \label{fig:prediction_residuals}
\end{figure}

\section{A benchmark example: the unbalanced disk}\label{sec:unbalanced}
We now consider the unbalanced disk system proposed in \cite{Verhoek:25a} as a second example to showcase the performance of LPV $\gamma$-DDPC in controlling nonlinear systems. The discretized dynamics of the system to be controlled are
\begin{equation*}
    \begin{aligned}
        \theta_{k+1} & = \theta_k + T_s \omega_k, \\
        \omega_{k+1} & = \left( 1 - \frac{T_s}{\tau_m} \right) \omega_k + T_s \frac{mgl}{J} \sin(\theta_k) + T_s \frac{K_m}{\tau_m} u_k,\\
        y_k & = \theta_k + e_k,
    \end{aligned}
\end{equation*}
where $\theta$~[rad], $\omega$~[rad/s] are the angular position and speed of the disk, respectively, while $u$~[V] is the control voltage and $e$ is a zero-mean, Gaussian, white measurement noise with variance $\sigma_e^2$. The parameters (as well as their values and measurement units) are listed in \tablename{~\ref{tab:disk_parameters}}. While being nonlinear, this system can be equivalently represented as an LPV system (see~\cite{Verhoek:21a}), by using the scheduling map $\phi(\theta)=\textrm{sinc}(\theta)$. Accordingly, the scheduling variable is bounded to the set $\mathbb{P} = [-0.22, 1]$.

To assess the performance of LPV $\gamma$-DDPC, we perform 100 Monte Carlo simulations for data collection within the same setting used in \cite{Verhoek:25a}. In particular, for each Monte Carlo run, we collect $N_{\textrm{data}} = 89$ samples by applying a uniformly distributed input $u_k \sim \mathcal{U}(\mathbb{U})$ from an initial condition $\theta_0 = -\pi/4$, and $\omega_0 = 5$. For all tested predictive control schemes, we set the prediction horizon to $T = 20$, while we set $Q = 16$ and $R = 0.01$, and constrain the inputs and outputs of the system in the set $\mathbb{U} = [-10, 10]$ and $\mathbb{Y} = [-\pi, \pi]$, respectively.\\
In addition, for the LPV $\gamma$-DDPC controller, we select the regularization coefficients $\beta_2 = 0.03$ and $\beta_3=0$, while choosing $h_Z = h_U = 3$, $n_{Z_P} = 10$, and $n_{U_F} = 28$ following the heuristic on the prediction residuals decrease explained in the previous Section. Once again, while we control the system, we assume the future scheduling is constant, that is, we apply the GS approach.

\begin{table}[!tb]
        \caption{Parameters of the unbalanced disk, with $g=9.81$ [m/s$^2$] being the gravitational acceleration.}
    \label{tab:disk_parameters}
    \centering
    \begin{tabular}{|l|c|c|c|c|c|c|c}
        \hline
        \!\textbf{\!Param.\!}\! & \hspace*{-.1cm}$T_s$\hspace*{-.1cm} & \hspace*{-.1cm}$m$\hspace*{-.1cm} & \hspace*{-.1cm}$l$\hspace*{-.1cm} & \hspace*{-.1cm}$J$\hspace*{-.1cm} & \hspace*{-.1cm}$\tau_m$\hspace*{-.1cm} & \hspace*{-.1cm}$K_m$\hspace*{-.1cm} \\
        \hline 
        \!\!\textbf{Value} & \!0.01\! & \!0.076\! & \!0.041\! & \!$2.4\!\cdot\! 10^{-4}$\! & \!0.4\! & \!11\! \\
        \hline
        \!\!\textbf{Unit} & [s] & [kg] & [m] & [kg m\textsuperscript{2}] & [-] & [-] \\
        \hline
    \end{tabular}
\end{table}

\subsection{Comparison with the LPV-IO-DPC scheme} \label{subsec:comparison_with_LPV-IO-DPC}
We first study the performance of LPV $\gamma$-DDPC compared to the LPV-IO-DPC scheme\footnote{We set the hyperparameters of LPV-IO-DPC to $\lambda_g = 1$ and $\lambda_\sigma = 10^9$.} of \cite{Verhoek:25a}, both benchmarked against an oracle LPV MPC scheme with access to the true LPV dynamics, which employs a Kalman filter with weights $Q_\textrm{KF}=\textrm{diag}(10^{-4}, 10)$ and $R_\textrm{KF}=10^{-2}$ to estimate $\theta_k$ and $\omega_k$, and (like the data-driven predictive control approches) assumes the future scheduling signal to be constant over the prediction horizon.

First, we set $M=2$ and $\sigma_e = 2.5 \cdot 10^{-3}$. While generating a large Signal-to-Noise Ratio (SNR) of 60~[dB], these choices allow us to make a fair initial comparison with LPV-IO-DPC~\cite{Verhoek:25a}. In this low-noise scenario, the results attained in closed-loop are shown in \figurename{~\ref{fig:comparison_small_noise}}. These closed-loop outputs show that LPV $\gamma$-DDPC yields a closer response to the oracle solution on average, and a smaller standard deviation among Monte Carlo runs compared to the LPV-IO-DPC scheme.

We then increase the measurement noise by setting $\sigma_e = 10^{-2}$ (SNR of 46 [dB]), performing a new set of 100 MC runs, imposing $\beta_2=3$ in \eqref{eq:reduced_gammaDDPC} and setting the hyperparameters of LPV-IO-DPC a posteriori to minimize the number of unstable closed-loop trajectories\footnote{Accordingly, $\lambda_g = 0.6$.}. Across the 100 Monte Carlo runs, LPV-IO-DPC generated 21 unstable trajectories while LPV $\gamma$-DDPC never destabilized the system in closed-loop. In addition, as clear from \figurename{~\ref{subfig:large_noise_small_M}}, the tracking performance of LPV $\gamma$-DDPC are substantially better in average and standard deviation than that achieved throughout the 79 stable closed-loop instances with LPV-IO-DPC.    

We then evaluate whether performance with this higher level of noise in the data could be improved by increasing $M$. To this end, we consider a larger dataset $N_\textrm{data} = 120$ to guarantee persistence of excitation, and set $M=4$. The results of 100 MC closed-loop simulations, obtained by setting\footnote{The hyperparameter $\lambda_g$ of LPV-IO-DPC is set to $\lambda_g = 3$.} $\beta_2 = 0.5$, are shown in \figurename{~\ref{subfig:large_noise_large_M}}. Clearly, even for a larger past horizon, LPV $\gamma$-DDPC achieves better tracking performance and is closer to the oracle solution than LPV-IO-DPC. This advantage in performance is paired with one in computational complexity, as showcased in \tablename{~\ref{tab:simulation_time}}. Indeed, larger $M$ and $N_\textrm{data}$ increase the computational complexity of LPV-IO-DPC but does not affect LPV $\gamma$-DDPC, whose complexity depends solely on $n_{Z_P}$ and $n_{U_F}$. 
This highlights how our approach enables the use of larger past horizons compared to other DDPC methods for LPV systems.

\begin{figure}[t]
    \centering
    \includegraphics[width=\linewidth]{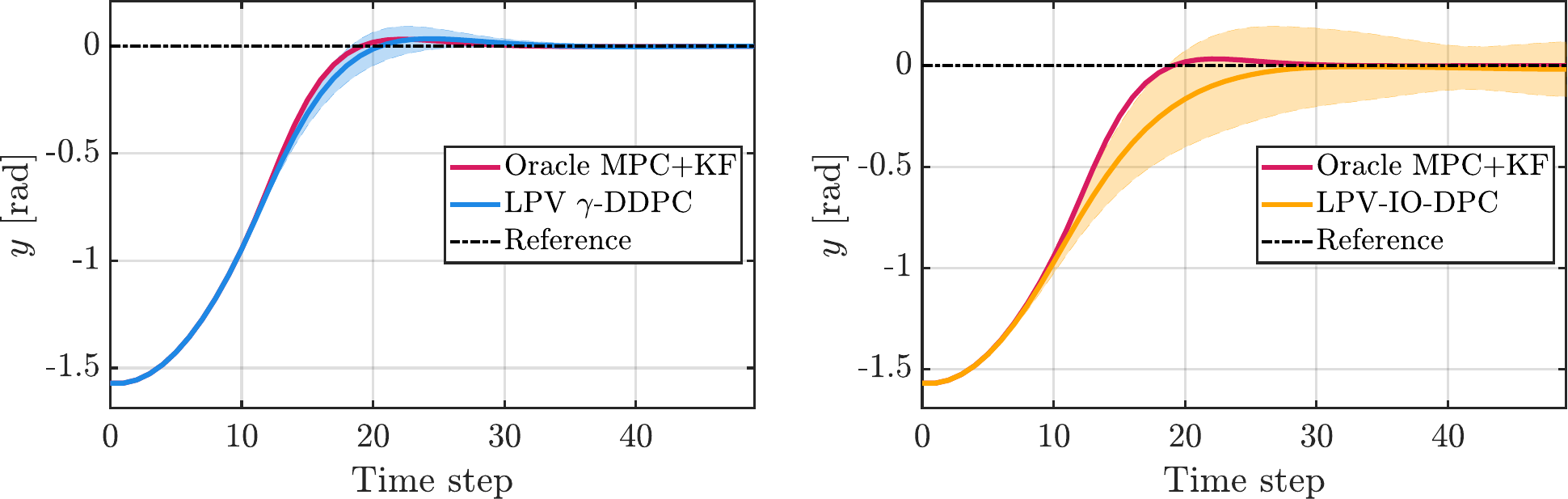}
    \caption{Unbalanced disk trajectories mean and standard deviation (shaded area) for a SNR of 60 [dB].}
    \label{fig:comparison_small_noise}
\end{figure}

\begin{figure}[t]
    \centering
    \begin{subfigure}{\linewidth}
        \centering
        \includegraphics[width=\linewidth]{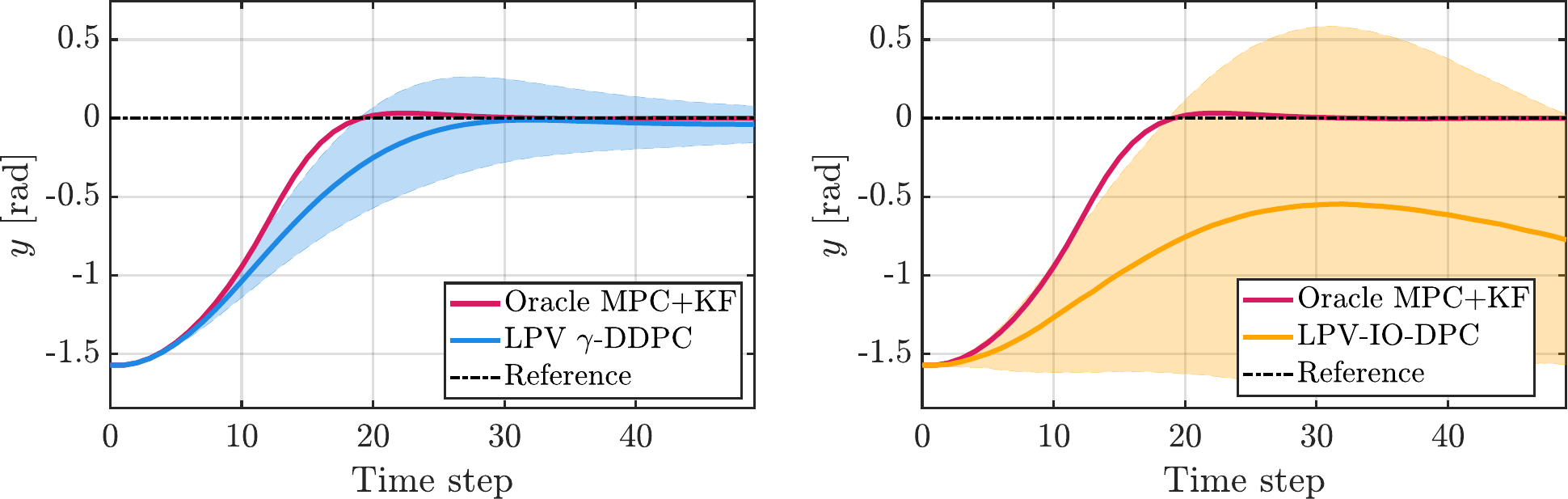}
        \caption{$N_\textrm{data}=89$, $M=2$}
        \label{subfig:large_noise_small_M}
    \end{subfigure}\\[2mm]% 
    \begin{subfigure}{\linewidth}
        \centering
        \includegraphics[width=\linewidth]{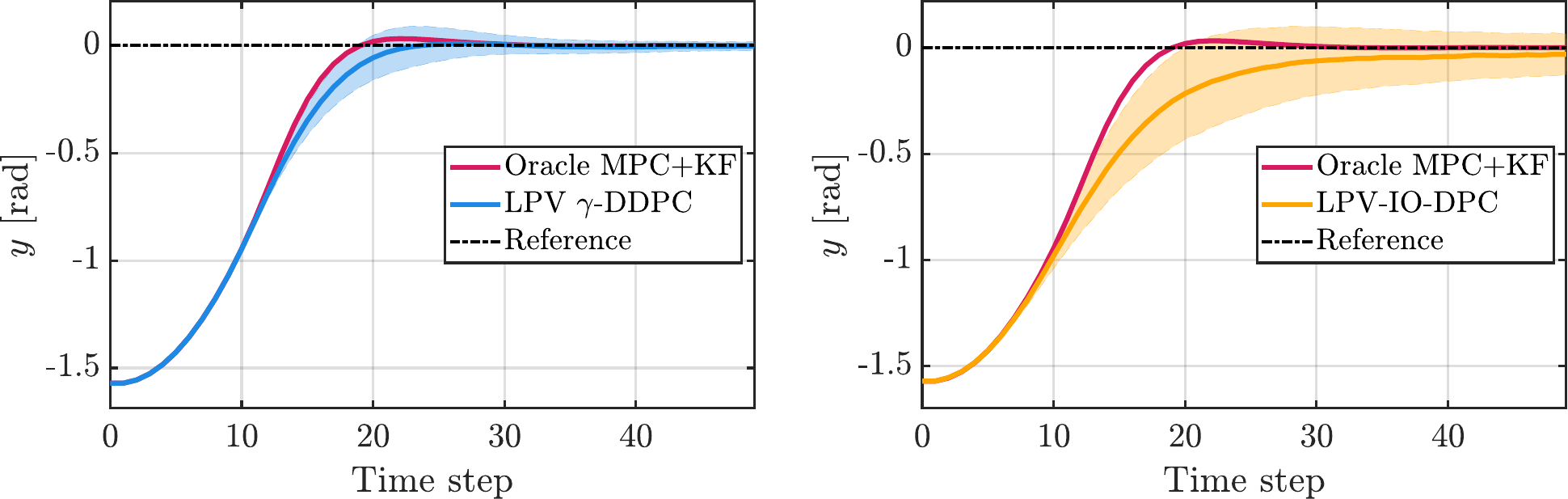}
        \caption{$N_\textrm{data}=120$, $M=4$}
        \label{subfig:large_noise_large_M}
    \end{subfigure}
    \caption{Unbalanced disk trajectories mean and standard deviation (shaded area) for a SNR of 46 [dB]. The LPV-IO-DPC unstable trajectories (21 out of 100) are not shown.}
    \label{fig:comparison_medium_noise}
\end{figure}

\begin{table}[!tb]
        \caption{Average simulation time [s] per Monte Carlo run \emph{vs} approach and dataset length.}
    \label{tab:simulation_time}
    \centering
    \begin{tabular}{|c|c|c|}
        \hline
        $\bm{N_\textrm{data}}$ & \textbf{LPV }$\bm{\gamma}$\textbf{-DDPC} & \textbf{LPV-IO-DPC} \\
        \hline 
         89  & 8.45 & 35.35 \\
         \hline
         120 & 8.52 & 73.97\\
         \hline
    \end{tabular}
\end{table}

\subsection{Robustness to inaccurate choice of the scheduling map}
While previously we have assumed to know a priori the correct scheduling map, this hypothesis implies some prior knowledge on the system dynamics, which is often unavailable in a data-driven setting. We thus continue our assessment of LPV $\gamma$-DDPC performance, analyzing its performance when using the wrong mapping $\phi(\theta) = \cos(\theta)$. To this end, we collect $N_{\textrm{data}} = 2000$ samples with a control input $u_k \sim \mathcal{U}(\mathbb{U})$, setting $M = 15$, $T = 10$, $Q = 16$, $R = 0.01$, $h_Z = h_U = 3$, $n_{Z_P} = 50$, and $n_{U_F} = 15$. By performing $30$ Monte Carlo simulations, we first set $\sigma_{e} = 0.01$, yielding a SNR of 46~[dB]. In this case, we set $\beta_2 = 0.01$. We then increase the level of noise up to 23~[dB] by setting $\sigma_e = 0.1$, for which we take $\beta_2 = 1$. 

The attained tracking results are benchmarked against those obtained with the right scheduling map in \figurename{~\ref{fig:wrong_prior}}. In both noise scenarios, the wrong scheduling assumption generates some offset in reference tracking. Note that this offset could still be removed by introducing an integral action in the controller, as in \cite{Lazar:22a}, which we leave as future extension. Besides the tracking offset, even with the wrong prior on the scheduling map, LPV $\gamma$-DDPC allows the system output to closely follow the changes in the reference to be tracked. 

In addition, \figurename{~\ref{subfig:wrong_large_noise}} further highlights that LPV $\gamma$-DDPC can control the system, with an accuracy linked to the knowledge of the true scheduling map, even with a relatively small SNR. Note that this is not achievable with LPV-IO-DPC and other data-driven predictive control architectures for LPV systems that, in the same noisy scenario, do not enable the system to complete the tracking task.

\begin{figure}[!tb]
    \centering
    \begin{subfigure}{\linewidth}
        \centering
        \includegraphics[width=\linewidth]{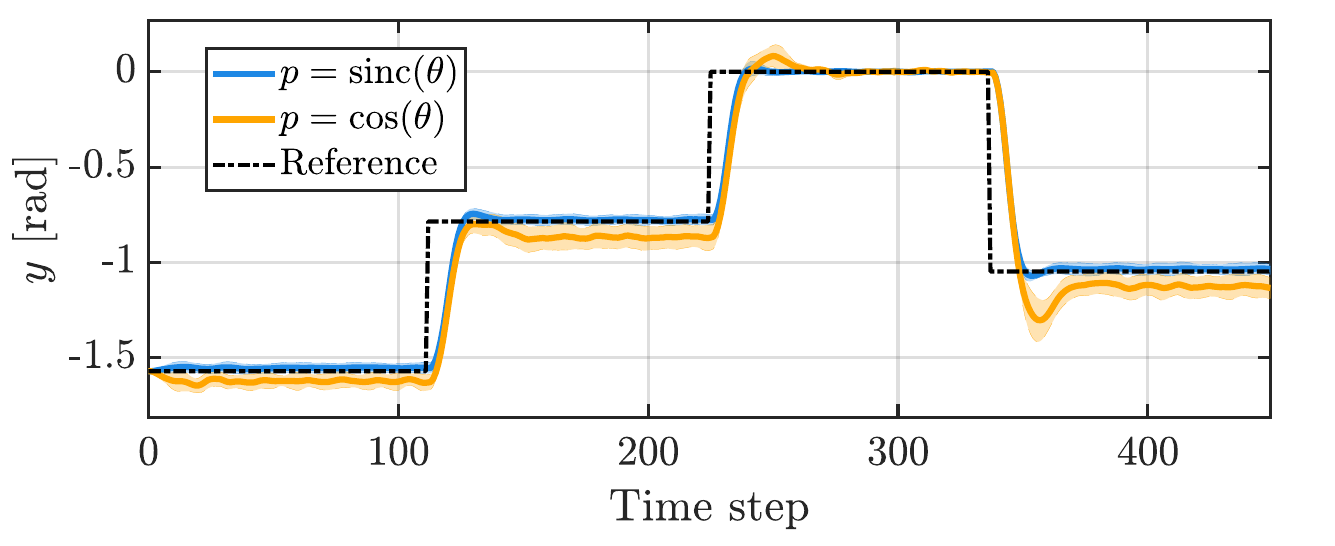}
        \caption{SNR of 46 [dB].}
        \label{subfig:wrong_small_noise}
    \end{subfigure}\\[2mm]% 
    \begin{subfigure}{\linewidth}
        \centering
        \includegraphics[width=\linewidth]{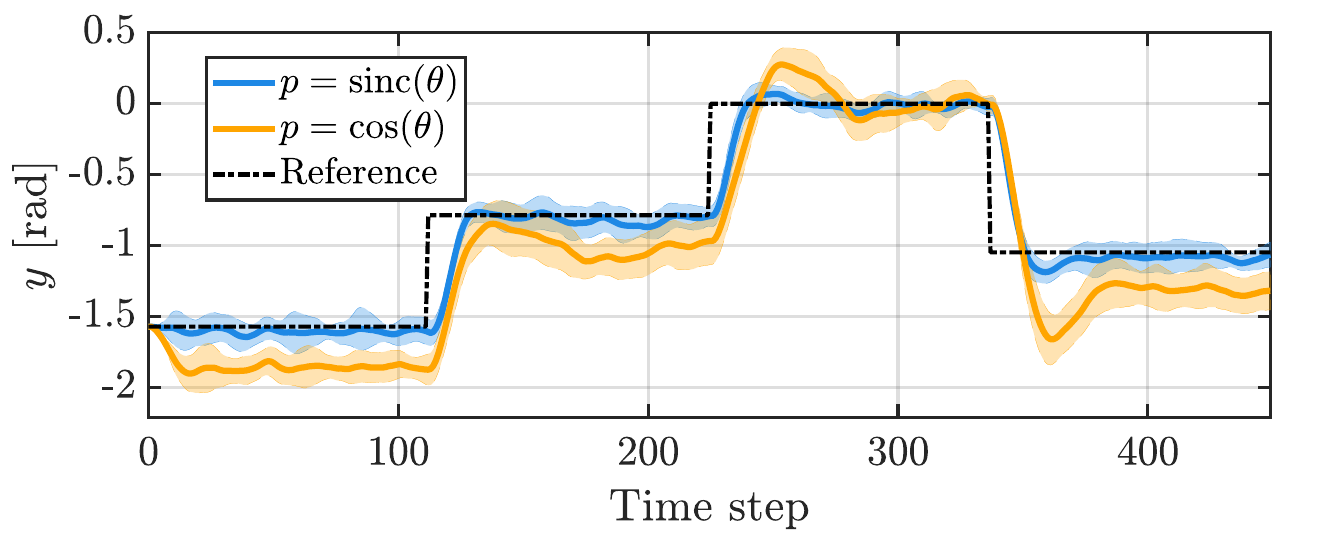}
        \caption{SNR of 23 [dB].}
        \label{subfig:wrong_large_noise}
    \end{subfigure}
    \caption{Mean (solid line) and standard deviation (shaded area) of the closed-loop output obtained by using the true map $\phi(\theta) = \textrm{sinc}(\theta)$ and the wrong map $\phi(\theta) = \cos(\theta)$.}
    \label{fig:wrong_prior}
\end{figure}

\section{Conclusions}\label{sec:conclusions}
In this paper, we propose a reduced-complexity LPV $\gamma$-DDPC, a computationally efficient, subspace-inspired, data-driven predictive control strategy. To achieve tractability of the controller, we proposed constructing a reduced-order data-driven predictor that trades off accuracy for computational complexity. This reduced-complexity predictor is one of the main contributions of this work, making the proposed data-driven control approach computationally more efficient than existing LPV data-driven control techniques. In our numerical examples, we  show that our approach is more resilient to measurement noise and less sensitive to the selection of control hyperparameters than other existing techniques for LPV data-driven control. The complexity reduction, in turn, enables the use of longer past horizons in the predictor, which are often desirable in predictive control~\cite{Berberich:20a,Breschi:23c}.

Future work will be devoted to the automatic selection of control hyperparameters and to extensions to robustify the current control structure, e.g., by using (nonlinear) velocity forms for robust tracking, and terminal constraints for recursive stability.

% Acknowledgements
\begin{ack}
We thank the authors of \cite{BouHamdan:24a} for sharing their code implementation.
\end{ack}

% Bibliography
\bibliographystyle{ieeetr} % <- nicer imo (always initials use etc.)
\bibliography{autosam.bib}

\end{document}